\documentclass[11pt]{article}

\usepackage[margin=1.25in]{geometry}
\usepackage{setspace}

\usepackage{etoolbox}
\makeatletter
\patchcmd{\abstract}{\small}{\normalsize\singlespacing}
  {}{\PackageError{ETformat}{Abstract formatting patch failed}{}}
\makeatother

\usepackage{amsmath}     
\usepackage{amssymb}     
\usepackage{amsthm}      
\usepackage{mathtools}   

\usepackage{booktabs}    
\usepackage{graphicx}
\usepackage{caption}
\usepackage{float}
\usepackage{placeins}
\usepackage{xcolor}

\definecolor{revisionblue}{RGB}{0,76,153}

\colorlet{revisioncolor}{black}

\newcommand{\rev}[1]{{\color{revisioncolor}#1}}
\newenvironment{revision}{\color{revisioncolor}}{}

\colorlet{revisioncolorTwo}{black}

\newcommand{\revTwo}[1]{{\color{revisioncolorTwo}#1}}

\usepackage[
    colorlinks=true,
    linkcolor=blue!45!black,
    citecolor=blue!45!black,
    urlcolor=blue!45!black
]{hyperref}

\usepackage[natbibapa]{apacite}
\AtBeginDocument{
    \renewcommand{\doi}[1]{\url{https://doi.org/#1}}
}

\usepackage{multirow}        
\usepackage{threeparttable}  
\usepackage{siunitx}         
\usepackage{rotating}

\theoremstyle{plain}                        
\newtheorem{theorem}{Theorem}[section]
\newtheorem{lemma}[theorem]{Lemma}
\newtheorem{appendixlemma}{Lemma}[section]
\newtheorem{proposition}[theorem]{Proposition}
\newtheorem{corollary}[theorem]{Corollary}

\theoremstyle{definition}                   

\newtheorem{assumption}[theorem]{Assumption}

\theoremstyle{remark}                       
\newtheorem{remark}[theorem]{Remark}

\newcommand{\wto}{\Rightarrow}                 

\newcommand{\eps}{\varepsilon}
\newcommand{\cbar}{\bar{c}}
\newcommand{\abar}{\bar{\alpha}}
\newcommand{\ty}{\tilde{y}}                    

\newcommand{\E}{\mathbb{E}}
\newcommand{\Op}{O_{\IP}}
\newcommand{\op}{o_{\IP}}

\newcommand{\IP}{\mathbb{P}}                   
\newcommand{\MZa}{M\!Z_{\alpha}}
\newcommand{\MZt}{M\!Z_{t}}
\newcommand{\MSB}{M\!S\!B}
\newcommand{\MP}{M\!P_{T}}

\begin{document}
\singlespacing

\title{A Characterization of the $M$-tests Under Nearly Integrated Nearly White Noise}
\author{Aidan Wardak\footnote{Email: wardakaidan@ufl.edu.}  and Sayar Karmakar\footnote{ Corresponding author. Email:sayarkarmakar@ufl.edu. The authors declare no competing interests.}}
\date{September 24, 2026}

\maketitle

\begin{abstract}
\normalsize\singlespacing
\noindent
We derive the limiting distributions of the $M$-test family of unit root statistics in the nearly integrated nearly white noise (NINW) framework introduced by \citet{NP1994} in the case of an unknown linear time trend. In the case of known long run variance (LRV), the limiting distributions of the $M^{GLS}$ tests are contaminated by additional noise terms as a result of quasi differencing whereas these terms are less present in the $M^{OLS}$ limiting distributions, both of which display conservative properties under conventional critical values. Furthermore, we prove the Gaussian power envelope in the NINW model is asymptotically equivalent to the standard envelope of \citet{ERS1996}, and that the oracle $M$-tests have inefficient power relative to this benchmark. We then derive the limiting distributions of the feasible statistics and show that the autoregressive estimate of the LRV commonly used overestimates the LRV, creating altered limiting distributions. Finally, finite sample simulations illustrate that, of the procedures considered, no uniformly satisfactory solution exists for handling a series with a large negative moving average coefficient.
\end{abstract}


\section{Introduction}
Testing for the presence of a unit root is common practice, but the reliability of many unit root tests drastically deteriorates under data-generating processes with a moving average coefficient near $-1$. The original Dickey-Fuller test \citep{DickeyFuller1979} was extended to hopefully accommodate an error process such as this by \citet{SaidDickey1984}. \citet{PP1988} developed a nonparametric correction to these tests based on estimation of the LRV (the spectral density at frequency zero). Yet, these tests display severe size distortions when the moving average coefficient is near $-1$, as documented in \citet{Schwert1989}, \citet{Agiakloglou1992EMPIRICALEO}, \citet{Leybourne1999}, and more. The $M$-test family of statistics originating with \citet{Stock1990} and further developed by \citet{NgPerron1996}, were specifically designed to improve robustness in precisely this type of data generating process (DGP). \citet{NP1998} subsequently analyzed the autoregressive LRV estimator used to implement these tests and found substantial size improvement over kernel-based methods like \citet{NeweyWest1987}.

\citet{NgPerron2001} combined the typical $M$-test construction with the GLS detrending developed by \citet{DuFourKing1991} and \citet{ERS1996} to further improve the power of the $M$-tests. Under the standard local-to-unity framework, the resulting $M^{GLS}$ tests display local asymptotic power functions close to that of the Gaussian power envelope while retaining more favorable size properties than conventional tests under a large negative moving average coefficient. Ng and Perron also emphasize that implementation of unit root tests depends greatly on the lag truncation parameter $k$ \citep{Agiakloglou1996,NgPerron1995,LOPEZ1997}. They find when the moving average coefficient is large and negative, a relatively large autoregression is necessary to control size, while conventional criteria like the AIC and BIC tend to select too few lags. The modified information criterion of \citet{NgPerron2001} was designed to address this problem. Furthermore, \citet{PQ2007} suggest using OLS detrended data to construct the modified information criterion to solve the problem of power reversal; that is, the power of GLS-based tests decrease when the autoregressive coefficient strays further from unity as documented in \citet{seo2005improving}.

The present paper revisits these results under the nearly integrated nearly white noise local asymptotic framework of \citet{NP1994} and \citet{NgPerron1996}. In this framework, the autoregressive coefficient approaches unity and the moving average coefficient approaches $-1$. As a result, the LRV collapses to 0. This degeneracy is the defining feature of this framework; quantities that are otherwise asymptotically negligible under the standard local-to-unity framework can become first order when measured relative to the vanishing LRV. \citet{NgPerron1996}, \citet{NP1998}, and \citet{NgPerron2001} analyze aspects of the $M$-tests and lag selection under this framework, but do not fully characterize the behavior of the $M$-tests with an unknown linear time trend, the associated Gaussian power bound, and the exact relative behavior of the feasible LRV estimator jointly. 

The first contribution of this paper is to derive the known LRV (oracle) distributions of the OLS and GLS based $M$-tests within the NINW framework. For OLS detrending, the resulting limiting distributions are a direct extension of the no-deterministics limiting distributions already derived by \citet{NgPerron1996}. However, the $M^{GLS}$ tests have limiting distributions that gain extra nuisance terms that remain asymptotically non-negligible as a result of quasi-differencing. The GLS limiting distributions therefore cannot be obtained by merely replacing the undetrended Ornstein-Uhlenbeck process with its standard GLS-detrended counterpart. The resulting OLS and GLS statistics display $\delta$ dependent support restrictions; for sufficiently small $\delta$, standard critical values lie outside these support bounds and the asymptotic rejection probability is hence exactly zero. Thus, the NINW problem can produce severe under-rejection even when the LRV is known.

The second contribution of this paper is to characterize the Gaussian power envelope under the NINW DGP. \citet{ERS1996} derive the standard Gaussian envelope under conditions that include a positive spectral density at frequency zero which fails under our framework. We find that the power envelope remains unchanged from the standard case and is hence not dependent on $\delta$. For every fixed $\delta >0$, exact whitening removes the nearly noninvertible MA component and the nuisance projection terms generated by the transformed intercept and trend either cancel between the null and point alternative or are offset by their normalization. Hence, the poor asymptotic power of the $M$-tests relative to this benchmark is a direct result of the inefficiency of the statistics, not an intrinsically harder testing problem. 

The third contribution concerns the feasible implementation of the tests in which the LRV must be estimated. \citet{NP1998} prove that $s^2_{AR} \stackrel{\IP}{\to} 0$ under NINW. \citet{NgPerron2001} subsequently show that, under NINW, the conditions $k^2 s^{2}_{AR} = O_\IP(1)$ and $Ts^2_{AR} = O_\IP(1)$ can both hold only at the lag rate $k/\sqrt{T} \to \kappa \in (0,\infty)$. Neither result, however, determines whether $s^2_{AR}$ is ratio consistent with the LRV, which is necessary for oracle and feasible limiting distribution equivalence as the LRV collapses to zero. We show that at this lag rate, the ratio of $s^2_{AR}$ and the true LRV approaches $\coth^2(\kappa\delta/2)$ asymptotically, and hence, $s^2_{AR}$ is not ratio consistent. This sharpens the earlier NINW consistency results by identifying the multiplicative distortion that remains at the rate required to keep the feasible tests bounded.

The asymptotic analysis is complemented by finite sample Monte Carlo experiments designed to examine the interaction between detrending, lag selection, and the admissible upper bound on $k$. Increasing the conventional lag cap to allow orders proportional to $\sqrt{T}$ can materially reduce the most extreme NINW oversizing, but this improvement comes with substantial losses in size-adjusted power. Furthermore, GLS based construction of the modified information criterion results in favorable size properties relative to an OLS based construction, but then the tests are subject to the power reversal problem. Finally, the $M^{GLS}$  and $M^{OLS}$ tests display similar size-adjusted power. These results illustrate a persistent size-power tradeoff rather than any uniformly satisfactory implementation.

The remainder of the paper is organized as follows. Section \ref{sec:model} introduces the NINW model, GLS detrending, and the statistics. Section \ref{sec:oracle} derives the oracle limiting distributions of the $M$-tests under NINW. Section \ref{sec:env} characterizes the Gaussian power envelope under NINW and analyzes the asymptotic power of the oracle $M$-tests. Section \ref{sec:feasible_ld} derives the feasible limiting distributions of the $M$-tests under NINW. Section \ref{sec:finite_sample} reports finite sample results. Section \ref{sec:conclusion} concludes. Proofs are relegated to the appendices.

\section{Model, GLS detrending, and statistics}\label{sec:model}

\subsection{Nearly integrated nearly white noise}
\begin{revision}
The data-generating process used throughout the paper is
\begin{equation}\label{eq:model}
\begin{aligned}
 y_t&=\beta_0+\beta_1t+u_t,\\
 u_t&=\rho_T u_{t-1}+v_t,\qquad \rho_T=1+\frac{c}{T},\\
 v_t&=e_t+\theta_Te_{t-1},\qquad \theta_T=-1+\frac{\delta}{\sqrt T},
 \qquad t=1,\ldots,T.
\end{aligned}
\end{equation}
The stochastic conditions used to derive the asymptotic results are collected next.

\begin{assumption}\label{ass:innovations}
The initialization is $u_0=e_0=0$. The innovations $\{e_t\}_{t=1}^T$ are i.i.d. with
\[
 \E e_t=0,\qquad \E e_t^2=\sigma_e^2\in(0,\infty),\qquad \E|e_t|^4<\infty.
\]
The local parameters satisfy fixed $c\leq0$ and fixed $\delta>0$ as $T\to\infty$. Whenever GLS detrending is used, $\cbar<0$ is fixed. Note that Gaussianity is not imposed here.
\end{assumption}

Observe the autoregressive coefficient approaches one at rate $T^{-1}$, while the moving-average coefficient approaches $-1$ at rate $T^{-1/2}$ (equivalently, the corresponding MA root approaches $+1$). This is the classical NINW model of \citet{NP1994} and \citet{NgPerron1996}, augmented here by an intercept and linear trend. Its defining feature is
\begin{equation}
    \omega^2_T = \sigma^2_e(1+\theta_T)^2 = \frac{\sigma^2_e \delta^2}{T}, \quad\text{so} \quad T\omega^2_T = \sigma^2_e\delta^2. \label{eq:lrv}
\end{equation}
Hence the LRV of $v_t$ tends to zero at rate $T^{-1}$.
\end{revision}

\subsection{GLS detrending}
Let $z_t = (1,t)'$ denote the deterministic regressors. Fix $\cbar<0$ and define $\abar = 1+\cbar/T$. For any series $y_t$, define its quasi-differenced counterpart by
\begin{align}
    y_1^{\abar} = y_1, \qquad y_t^{\abar} = y_t - \abar y_{t-1}, \quad t \geq 2, \label{eq:quasi_w}\\
    z_1^{\abar} = z_1, \qquad z_t^{\abar} = z_t - \abar z_{t-1}, \quad t \geq 2. \label{eq:quasi_z}
\end{align}
The GLS estimator of the deterministic coefficients is then
\begin{equation}
    \hat{\psi} = \arg \min_\psi \sum_{t=1}^T (y^{\abar}_t - \psi'z_t^{\abar})^2, \label{eq:gls_minimizer}
\end{equation}
so that the GLS detrended series is $\ty_t = y_t - \hat{\psi}'z_t.$ \citet{ERS1996} recommended $\cbar = -13.5$ for an unknown linear time trend. Observe that the first observation is left undifferenced; this fact is important for the limiting distributions of the statistics.

\subsection{The statistics}\label{sec:stats}
The $M^{GLS}$ family of statistics is defined as 
\begin{align}
    &\MZa^{GLS} = \frac{T^{-1}\ty^2_T - s^2_{AR}}{2T^{-2}\sum_{t=2}^T\ty^2_{t-1}}, \qquad \MSB^{GLS} = \left(\frac{T^{-2}\sum_{t=2}^T\ty^2_{t-1}}{s^2_{AR}} \right)^{1/2}, \label{eq:mza_msb}\\
    &\MZt^{GLS} = \MZa^{GLS} \cdot \MSB^{GLS}, \qquad \MP^{GLS} = \frac{\cbar^2 T^{-2} \sum_{t=2}^T\ty^2_{t-1} + (1-\cbar)T^{-1}\ty^2_T}{s^2_{AR}}, \label{eq:mzt_mpt}
\end{align}
where $s^2_{AR}$ is an autoregressive estimate of the LRV to be handled in more detail later. For their OLS counterparts, replace $\ty$ in the level moments by OLS detrended $y$; the common estimator $s^2_{AR}$ is still constructed from GLS residuals. For the following section, we let $s^2_{AR}$ be replaced by the true LRV $\omega_T^2$ to understand the behavior of the statistics when they are gifted the true LRV.

\section{Oracle limiting distributions}\label{sec:oracle}

\subsection{Limit objects}
Before expressing the limiting distributions, it is important to define some terms that will appear in the distributions. Let $W$ denote standard Brownian motion on [0,1] and, for $c \in \mathbb{R}$, let 
\begin{align}
    J_c(r) = \int_0^r e^{(r-s)c} dW(s) \label{eq:jcr}\\
    J_c^\tau(r) = J_c(r) - \hat{a} - \hat{b}r, \label{eq:jcr_tau}
\end{align}
where $\hat{b} = 12 \int_0^1(s-\frac{1}{2})J_c(s)ds$, and $\hat{a} = \int_0^1J_c(s)ds - \frac{1}{2}\hat{b}$. 

Let $\varepsilon_1$ and $\varepsilon_{\infty}$ denote the limiting variables associated with the standardized first and last sample innovations $e_1/\sigma_e$ and $e_T/\sigma_e$, respectively, as constructed in Lemma \ref{A.0}. They are independent draws from the distribution of $e_t/\sigma_e$ and are independent of the Brownian motion generating $J_c$. Note that our $\eps_\infty $ is equivalent to $e_\infty$ in \citet[Thm.~3.1]{NgPerron1996}.

The terms generated by $\cbar$ are
\begin{equation}
    \lambda = \frac{1-\cbar}{1-\cbar + \cbar^2/3}, \qquad q_{\cbar} = \lambda\eps_{\infty} - \frac{3-\lambda}{2}\eps_1. \label{eq:limit_constants}
\end{equation}
Finally, the typical GLS functional
\begin{equation}
    V_{c,\cbar}(r) = J_c(r) - r \left( \lambda J_c(1) + 3(1-\lambda) \int_0^1 sJ_c(s)ds \right). \label{gls_func}
\end{equation}

\subsection{Limiting distributions of sample moments}\label{sec:lim_of_sm}

We first derive the limiting distributions of the two sample moments fundamental to the $M$-test family. Note that $\wto$ denotes convergence in distribution throughout the paper.
\begin{proposition}[Sample moment limits under OLS]\label{prop:sm_ols} Let $\{y_t\}$ be defined by \eqref{eq:model} and Assumption~\ref{ass:innovations} hold. For OLS-detrended data,
\begin{align}
    &T^{-1}\sum_{t=2}^T \ty^2_{t-1} \wto \sigma^2_e\left(1+\delta^2 \int_0^1 J_c^{\tau}(r)^2 dr \right) \label{eq:sm1_ols} \\
    &\ty_T \wto \sigma_e(\eps_{\infty} + \delta J_c^{\tau}(1)). \label{eq:sm2_ols}
\end{align}

\end{proposition}

\begin{proposition}[Sample moment limits under GLS]\label{prop:sm_gls} Let $\{y_t\}$ be defined by \eqref{eq:model} and Assumption~\ref{ass:innovations} hold, with fixed $\cbar<0$. For GLS-detrended data,

\begin{align}
    &T^{-1}\sum_{t=2}^T \ty^2_{t-1} \wto \sigma^2_e \left[1+\int_0^1(\delta V_{c,\cbar}(r) - \eps_1 - q_{\cbar}r)^2dr \right] \label{eq:sm1_gls} \\
    &\ty_T \wto \sigma_e \Big[(1-\lambda)\eps_{\infty} + \frac{1-\lambda}{2}\eps_1 + \delta V_{c, \cbar}(1)  \Big]. \label{eq:sm2_gls}
\end{align}    
\end{proposition}

\begin{remark}
As can be seen in the previous propositions, we can simply replace the undetrended process $J_c$, with the corresponding detrended process $J_c^{\tau}$ in the no-deterministics sample moment limits of \citet{NgPerron1996} to obtain the limits under an unknown linear time trend. This works as the OLS fitted line to the innovations is asymptotically negligible, so only the fitted line to the near-integrated component survives in addition to $\varepsilon_\infty$. On the other hand, under GLS detrending, the fitted line to the innovations survives asymptotically, introducing additional nuisance terms to the sample moment limits. Consequently, the GLS limits cannot be obtained from the OLS limits by simply replacing $J_c^{\tau}$ with its GLS detrended counterpart $V_{c,\cbar}.$
\end{remark}

\subsection{Limiting distributions of statistics}\label{sec:limit_distributions}

\begin{theorem}[Oracle $M^{OLS}$ limits under NINW]\label{thm:oracle1} Let $\{y_t\}$ be defined by \eqref{eq:model} and Assumption~\ref{ass:innovations} hold. For OLS-detrended data and known LRV,
\begin{align}
    \MZa^{OLS} &\wto \frac{(\eps_{\infty}+\delta J_c^\tau(1))^2 - \delta^2}{2 \left(1+\delta^2 \int_0^1J_c^{\tau}(r)^2dr\right)}, \notag\\
    \MSB^{OLS} &\wto \frac{1}{\delta}\left(1+\delta^2 \int_0^1 J_c^{\tau}(r)^2dr\right)^{1/2} ,\label{eq:mza_msb_limt_ols} \\
    \MZt^{OLS} &\wto \frac{(\eps_{\infty}+\delta J_c^{\tau}(1))^2 - \delta^2}{2\delta \left(1 + \delta^2 \int_0^1 J_c^{\tau}(r)^2dr\right)^{1/2}} \label{eq:mzt_limit_ols}
\end{align}
These limiting distributions are precisely the limiting distributions derived by \citet{NgPerron1996} in the no-deterministics case, except we replace $J_c(r)$ with $J_c^{\tau}(r)$ whenever the former appears.
\end{theorem}
\begin{revision}
\begin{proof}
Let $D_T^{OLS}=T^{-1}\sum_{t=2}^T\ty_{t-1}^2$. Since $T\omega_T^2=\sigma_e^2\delta^2$, the definitions in \eqref{eq:mza_msb} imply
\[
 \MZa^{OLS}=\frac{\ty_T^2-T\omega_T^2}{2D_T^{OLS}},\qquad
 \MSB^{OLS}=\left(\frac{D_T^{OLS}}{T\omega_T^2}\right)^{1/2}.
\]
Apply Proposition~\ref{prop:sm_ols} and the continuous-mapping theorem; the limit for $\MZt^{OLS}$ follows from $\MZt^{OLS}=\MZa^{OLS}\MSB^{OLS}$.
\end{proof}
\end{revision}

\begin{theorem}[Oracle $M^{GLS}$ limits under NINW]\label{thm:oracle2} Let $\{y_t\}$ be defined by \eqref{eq:model} and Assumption~\ref{ass:innovations} hold, with fixed $\cbar<0$. For GLS-detrended data and known LRV,
\begin{align}
    \MZa^{GLS} &\wto \frac{[(1-\lambda)\eps_{\infty} + \frac{1-\lambda}{2}\eps_1+ \delta V_{c,\cbar}(1)]^2 - \delta^2}{2\left(1+ \int_0^1[\delta V_{c,\cbar}(r) - \eps_1 - q_{\cbar}r]^2dr\right)}, \label{eq:mza_limit_gls} \\
    \MSB^{GLS} &\wto \frac{1}{\delta}\left(1 + \int_0^1 [\delta V_{c,\cbar}(r) - \eps_1 - q_{\cbar}r]^2dr \right)^{1/2} \label{eq:msb_limit_gls} \\
    \MZt^{GLS} &\wto \frac{[(1-\lambda)\eps_{\infty} + \frac{1-\lambda}{2}\eps_1+ \delta V_{c,\cbar}(1)]^2 - \delta^2}{2\delta \left( 1 + \int_0^1[\delta V_{c,\cbar}(r) - \eps_1 - q_{\cbar}r]^2dr\right)^{1/2}}, \label{eq:mzt_limit_gls} \\
    \MP^{GLS} &\wto \frac{1}{\delta^2} \Bigg\{\cbar^2 \left(1 + \int_0^1[\delta V_{c,\cbar}(r) - \eps_1 - q_{\cbar}r]^2dr\right) \notag\\
    &\qquad{} + (1-\cbar) \left( (1-\lambda)\eps_{\infty} + \frac{1-\lambda}{2}\eps_1 + \delta V_{c,\cbar}(1) \right)^2 \Bigg\} \label{eq:mpt_limit_gls}
\end{align}
    
\end{theorem}
\begin{revision}
\begin{proof}
Set $D_T^{GLS}=T^{-1}\sum_{t=2}^T\ty_{t-1}^2$. Multiplying the numerator and denominator of each statistic in \eqref{eq:mza_msb}--\eqref{eq:mzt_mpt} by the appropriate power of $T$, and using $T\omega_T^2=\sigma_e^2\delta^2$, reduces every statistic to a continuous function of $(D_T^{GLS},\ty_T)$. Proposition~\ref{prop:sm_gls} and continuous mapping then give the four displayed limits.
\end{proof}
\end{revision}

\begin{remark}
    Clearly, the limiting distributions of the tests, particularly the GLS based tests, carry extra nuisance terms not present in the standard limiting distributions of the tests. Yet,
    observe that as $\delta \rightarrow \infty$, we recover the standard limiting distributions of \citet{NgPerron2001} and \citet{Stock1990}. To see this, define
    \[
    D = 1+\int_0^1\left[ \delta V_{c,\cbar}(r) - \eps_1 - q_{\cbar}r \right]^2dr, \qquad E = (1-\lambda)\eps_{\infty} + \frac{1-\lambda}{2}\eps_1 + \delta V_{c,\cbar}(1).
    \]
    Then
    \[
    \frac{D}{\delta^2} \xrightarrow{\mathrm{a.s.}} \int_0^1 V_{c,\cbar}(r)^2dr, \qquad \frac{E}{\delta} \xrightarrow{\mathrm{a.s.}} V_{c,\cbar}(1).
    \]
    Consequently,
    \[
    \frac{E^2-\delta^2}{2D} \xrightarrow[\delta\to\infty]{\mathrm{a.s.}} \frac{V_{c,\cbar}(1)^2 - 1}{2\int_0^1 V_{c,\cbar}(r)^2dr},
    \]
    the standard limiting distribution of $\MZa^{GLS}$ in the case of an unknown linear time trend. Similar arguments yield the same results for the other $M$-tests.

    The opposite regime is nonuniform. As $\delta \rightarrow 0$, all dependence on the true local-to-unity parameter $c$ enters through the vanishing term $\delta V_{c, \cbar}$ or $\delta J_c^{\tau}$. The surviving leading terms are instead determined by the boundary terms $\eps_1$ and/or $\eps_{\infty}$ and the fixed GLS design parameter $\cbar$. Thus, the leading null and local alternative laws become indistinguishable as $\delta \rightarrow 0$. This is the partly the mechanism behind the deterioration in size-adjusted power examined later.
\end{remark}

\begin{remark}\label{rem:innovation_dep}
    Since the $M$-test limiting distributions contain the terms $\varepsilon_1$ and $\varepsilon_{\infty}$, they depend on the full distribution of the standardized innovations, $e_t/\sigma_e$, beyond their first two moments. Hence Gaussian NINW critical values need not be valid for other innovation laws. Note that $\MSB^{OLS}$ is the exception to this.
    
\end{remark}

\begin{corollary}[Support bounds]\label{cor:support_bounds}
Under the conditions of Theorems~\ref{thm:oracle1}--\ref{thm:oracle2}, for every fixed $c\leq0$ the oracle NINW limits satisfy
\begin{equation}
    \MZa \geq -\frac{\delta^2}{2}, \qquad
    \MSB \geq \frac{1}{\delta}, \qquad
    \MZt \geq -\frac{\delta}{2}, \qquad
    \MP^{GLS} \geq \frac{\cbar^2}{\delta^2} \quad\text{a.s.} \label{eq:oracle_bounds}
\end{equation}
Consequently, for a critical value $cv$ the asymptotic rejection probability is exactly zero whenever
\[
 \delta<\sqrt{2|cv|},\qquad
 \delta<\frac1{cv},\qquad
 \delta<2|cv|,\qquad
 \delta<\frac{|\cbar|}{\sqrt{cv}},
\]
respectively. Moreover, the finite sample rejection probability converges to $0$ as $T \to \infty.$
\end{corollary}

\begin{remark}
For the linear time trend case, \citet{NgPerron2001} report 5\% critical values of -17.3, 0.168, -2.91, and 5.48 for $\MZa^{GLS}, \MSB^{GLS}, \MZt^{GLS}$ and $\MP^{GLS}$, respectively. Corollary~\ref{cor:support_bounds} therefore gives
\[
\delta^*_{\MZa} \approx 5.88, \quad \delta^*_{\MSB} \approx 5.95, \quad \delta^*_{\MZt} \approx  5.82, \quad \delta^*_{\MP} \approx 5.77.
\]
 Hence, all four GLS-based 5\% size tests enter their zero asymptotic rejection probability region at approximately $\delta^* =$ 5.8--6.0. The magnitude of this region can be made more concrete by relating it to the negative moving average coefficient itself. Since
 \[
 \theta_T = -1 + \frac{\delta}{\sqrt{T}},
 \]
 at $T=100$, the thresholds above correspond to 
 \[
 \theta_T  \approx -0.412, \quad -0.405, \quad -0.418, \quad -0.423
 \]
 for $\MZa^{GLS}, \MSB^{GLS}, \MZt^{GLS}$, and $\MP^{GLS}$, respectively. Thus, for example, along a NINW sequence with $\delta<5.88$, we have that for
 \[
 \theta_T<-1+\frac{5.88}{\sqrt{T}},
 \]
 the size and local asymptotic power of $\MZa^{GLS}$ computed using the 5\% asymptotic critical values converge to zero. For the $M^{OLS}$ tests, $\delta^* \approx 6.40$; an even greater region of zero asymptotic rejection probability.
\end{remark}

\section{Asymptotic power of the tests}\label{sec:env}

\subsection{The asymptotic Gaussian power envelope}

The limiting distributions of Section~\ref{sec:oracle} do not by themselves
determine whether the behavior of the $M$-tests reflects a harder testing
problem or inefficiency of the statistics. To distinguish the two, we
characterize the maximal asymptotic power attainable in the NINW
experiment under an unknown linear time trend, following the construction of \citet{ERS1996}. Their envelope
is derived under the condition that $\{v_t\}$ has a strictly positive
spectral density at frequency zero (their Condition~A). By \eqref{eq:lrv}
the LRV $\omega_T^2=\sigma_e^2\delta^2/T$ vanishes in the
limit, so that condition fails and equality of the two power bounds cannot
be assumed.

Assume, just in this section, that in addition to Assumption \ref{ass:innovations}$,\{e_t\}$ is Gaussian and that
$\Sigma_T=\operatorname{Cov}(v_1,\dots,v_T)$, is available to the
researcher. Granting this knowledge can only raise attainable power, so the
resulting bound is valid for any feasible trend invariant procedure. Let
$y^{\abar},Z^{\abar}$ denote the data and trend regressors quasi-differenced
at $\abar$ as in \eqref{eq:quasi_w}--\eqref{eq:quasi_z}, and let
$y^{1},Z^{1}$ denote the same construction at $\abar=1$, that is, first
differences with the initial observation retained. Define the Gaussian
criteria
\begin{equation}
L_T(\abar,\psi)=\bigl(y^{\abar}-Z^{\abar}\psi\bigr)'\,\Sigma_T^{-1}\,
\bigl(y^{\abar}-Z^{\abar}\psi\bigr),
\qquad
L_T(1,\psi)=\bigl(y^{1}-Z^{1}\psi\bigr)'\,\Sigma_T^{-1}\,
\bigl(y^{1}-Z^{1}\psi\bigr),
\label{eq:ersobjective}
\end{equation}
each of which is, up to an additive constant, minus two times the Gaussian
log-likelihood under the alternative $c=\cbar$ and under the null $c=0$
(where $\abar=1$), respectively. The most powerful test of $H_0\colon c=0$
against the point alternative $c=\cbar$ that is invariant to the trend
coefficients rejects for small values of the likelihood-ratio statistic
\begin{equation}
L^{*,\tau}_T(\cbar)
=\min_{\psi}L_T(\abar,\psi)-\min_{\psi}L_T(1,\psi),
\label{eq:Lstar}
\end{equation}
the difference in weighted sums of squared residuals from two constrained
GLS regressions, one imposing the alternative and one imposing the null
\citep[eqn.~(6)]{ERS1996}. Because the asymptotically sufficient statistic for
$c$ is two-dimensional, no uniformly most powerful invariant test exists:
\eqref{eq:Lstar} defines an infinite family of point-optimal tests indexed
by $\cbar$, none dominating the others at every alternative.

Under standard errors satisfying Condition~A ($\{v_t\}$ has a positive spectral density at frequency zero), the rejection region of the
test indexed by $\cbar$ admits a continuous-time representation. With
$V_{c,\cbar}$ as in \eqref{gls_func}, the asymptotic power function of the
size-$\alpha$ test indexed by $\cbar$, evaluated at the true local
parameter $c$, is
\begin{equation}
\pi^{\tau}(c,\cbar;\alpha)
=\Pr\!\left[\;\cbar^{2}\!\int_0^1 V_{c,\cbar}(r)^2\,dr
+(1-\cbar)\,V_{c,\cbar}(1)^2 \,<\, b^{\tau}_{\alpha}(\cbar)\right],
\label{eq:pitau}
\end{equation}
where $b^{\tau}_{\alpha}(\cbar)$ is defined by the same probability
statement at $c=0$ equaling $\alpha$ \citep[eqn.~(8)]{ERS1996}. Note $c$ indexes the true local
alternative at which power is evaluated, while $\cbar$ indexes which member
of the point-optimal family is used. Since the test indexed by $\cbar$ is
most powerful invariant against $c=\cbar$, no invariant size-$\alpha$ test
can exceed $\pi^{\tau}(\cbar,\cbar;\alpha)$ at that alternative. The
pointwise upper bound on asymptotic power is therefore the diagonal,
\begin{equation}
\Pi^{\tau}_{\mathrm{ERS}}(c;\alpha)\equiv\pi^{\tau}(c,c;\alpha),
\label{eq:ers_envelope}
\end{equation}
the ERS linear-trend Gaussian power envelope.

Let $\Pi^{\tau}_{\mathrm{NINW}}(c,\delta;\alpha)$ denote the corresponding
object under \eqref{eq:model} and Assumption~\ref{ass:innovations}: the maximal
asymptotic power attainable at the fixed alternative $(c,\delta)$ by
size-$\alpha$ tests invariant to the intercept and linear trend. The
results of Section~\ref{sec:oracle} might suggest that this envelope is
degraded relative to \eqref{eq:ers_envelope} since the statistics' limits depend
on $\delta$, are contaminated by the boundary variates, and lose their
dependence on $c$ as $\delta\to0$. The following theorem shows that it is
not.

\begin{theorem}[Gaussian power envelope under NINW]\label{thm:power_envelope}
Let $\Pi^{\tau}_{NINW}(c,\delta;\alpha)$ be the Gaussian power envelope under \eqref{eq:model}, Assumption~\ref{ass:innovations}, and the Gaussian oracle information set above. Then, for every fixed $c<0$, $\delta>0$, and $0<\alpha<1$, we have
\begin{equation}
    \Pi^{\tau}_{NINW}(c, \delta;\alpha) = \Pi^{\tau}_{ERS}(c;\alpha). \label{eq:envelope_equivalence}
\end{equation}
Hence the asymptotic Gaussian power envelope under an unknown linear time trend is identical to the standard ERS power envelope under an unknown linear time trend, and moreover, it does not depend on $\delta$.
\end{theorem}
\begin{remark}
This theorem has a consequential interpretation. Although $\delta$ enters the limiting distributions of the NINW $M$-tests, it does not enter the Gaussian power envelope. Thus the NINW process does not offer unit root tests intrinsically less information than in the standard case. Hence any loss in power at lower $\delta$'s can be interpreted as a fault of the oracle $M$-tests, not the DGP. Note that attainability of the same envelope by a fully feasible procedure is a separate question.
\end{remark}

\subsection{Asymptotic power of the \texorpdfstring{$M$}{M}-tests}
Now that we have characterized the Gaussian power envelope under NINW, we examine the asymptotic power of the $M$-tests. Figure~\ref{fig:asymptotic_power_mpt} graphs the asymptotic power function of the $MP_T^{GLS}$ test at various $\delta$'s and the Gaussian power envelope. Power is simulated from the asymptotic limiting functionals using 100,000 null and 100,000 alternative Monte Carlo replications and a $50,000$ point discretization grid. The $\MP$ curves use $\cbar=-13.5$ and $\delta$ specific critical values; the Gaussian power envelope is simulated point-wise, $\cbar=c$, at each local alternative. Size is set to 0.05.

\FloatBarrier

\begin{figure}[t]
    \centering
    \includegraphics[width=\textwidth]{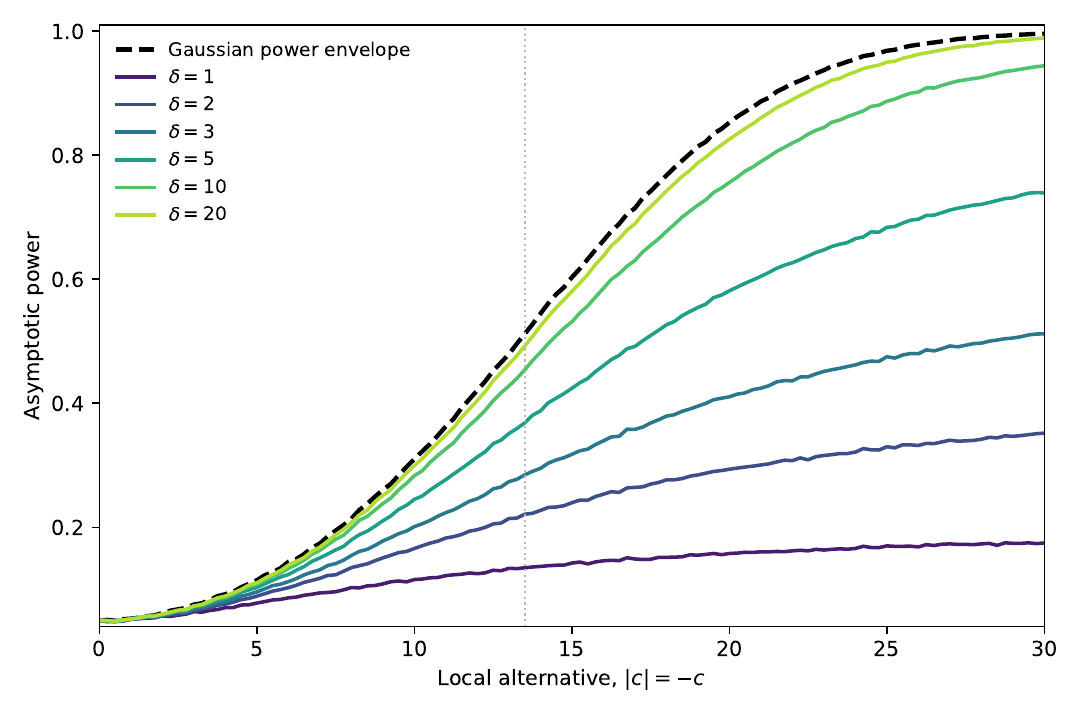}
    \caption{Asymptotic power of $\MP$ at various deltas compared to the Gaussian power envelope. The vertical line is $ c =-13.5.$}
    \label{fig:asymptotic_power_mpt}
\end{figure}

Figure~\ref{fig:asymptotic_power_mpt} illustrates that the power loss becomes increasingly severe as $\delta$ decreases. For $\delta=20$, the power curve is close to the envelope, consistent with the recovery of the standard limiting distribution as $\delta$ becomes large. At moderate to low values of $\delta$, however, a substantial loss of power exists. The inefficiency relative to the envelope is especially clear at $\delta=1$. Power remains below 0.2 as $c=-30$, whereas the envelope is almost at one.  

It is important to emphasize that these power losses occur when the LRV is treated as known. Therefore, these results precede any difficulties with LRV estimation under NINW. In fact, a feasible $M^{GLS}$ test that shares the same limiting distributions as the oracle case would only inherit the intrinsic power inefficiency problem. As a further robustness check, we reoptimized $\cbar$ at $\delta=1$ for the various tests and found that power was still well below the envelope, despite this completely infeasible advantage. Thus, we can conclude that the $M^{GLS}$ tests are intrinsically inefficient in terms of asymptotic power under NINW.

Note that similar results hold for the OLS variants of the $M$-tests, with one key exception, $\MSB^{OLS}$, which actually has a power function invariant to $\delta$. This was observed by \citet{NgPerron1996} in the no-deterministics case. The fact that the test's asymptotic power is invariant to $\delta$ can be derived from its limiting distribution. Let $A^{\tau}_c = \int_0^1J_c^{\tau}(r)^2dr$, then the limiting distribution from \eqref{eq:mza_msb_limt_ols} can be rewritten as $\MSB^{OLS} \wto (A^{\tau}_c + \delta^{-2})^{1/2}$. Then, the corresponding lower tail null critical value is $(a^{\tau}_{\alpha}+\delta^{-2})^{1/2}$, where $a^{\tau}_{\alpha}$ is the lower $\alpha$-quantile of $A^{\tau}_0.$ Hence, the size corrected rejection event is $A^{\tau}_c < a^{\tau}_{\alpha}$ because the common term $\delta^{-2}$ cancels. Thus, the asymptotic power of $\MSB^{OLS}$ is independent of $\delta$ and coincides with its standard local-to-unity power function. This property is unique to the $\MSB^{OLS}$ test.

\FloatBarrier

\section{Feasible limiting distributions}\label{sec:feasible_ld}
The entire analysis so far treated the LRV $\omega^2_T=\sigma^2_e\delta^2/T$ as known. Obviously, the practitioner does not implement these oracle tests, but rather a feasible test which must use an estimator for the LRV. Since $\omega^2_T \to 0$ as $T \to \infty$, the condition for feasible and oracle equivalence (of limiting distributions) is stronger than just absolute consistency.

To make this clear, let $\hat{\omega}^2_T$ be an estimator of $\omega^2_T$. Absolute consistency implies
\begin{equation}
    \hat{\omega}^2_T - \omega^2_T \stackrel{\IP}{\to} 0. \label{eq:abs_con}
\end{equation}
When $\omega^2_T$ is a fixed positive constant, \eqref{eq:abs_con} implies ratio consistency, but under NINW it does not. For example, suppose $\hat{\omega}^2_T=2\omega^2_T$. Then \eqref{eq:abs_con} is satisfied even though it remains twice the true LRV at all $T$. The relevant requirement is therefore
\begin{equation}
    \frac{\hat{\omega}^2_T}{\omega^2_T} \stackrel{\IP}{\to} 1, \quad\text{equivalently,} \quad T\hat{\omega}^2_T \stackrel{\IP}{\to} \sigma^2_e \delta^2. \label{eq:rel_con}
\end{equation}

We refer to \eqref{eq:rel_con} as ratio consistency. If it holds, then $\omega^2_T$ can be replaced by $\hat{\omega}^2_T$, leaving the oracle limiting distributions unchanged.

\subsection{The autoregressive long-run variance estimator}
We consider the autoregressive estimator used by \citet{NgPerron1996,NP1998} and \citet{NgPerron2001}. Throughout this subsection, $\ty_t$ denotes the GLS detrended series. The resulting $s^2_{AR}$ is used for both the OLS and GLS versions of the $M$-statistics; the distinction between them concerns the level moments, not the construction of the LRV estimator. For a truncation order $k$, estimate
\begin{equation}
    \Delta\ty_t = \hat{b}_0\ty_{t-1} + \sum_{j=1}^k \hat{b}_j\Delta\ty_{t-j} + \hat{e}_{t,k} \label{eq:ADF}
\end{equation}
by least squares over $t=k+2,\ldots,T$, and define
\begin{equation}
    \hat{b}(1)= \sum_{j=1}^k \hat{b}_j, \quad \hat{\sigma}^2_{ek}= \frac{1}{T}\sum_{t=k+2}^T \hat{e}^2_{t,k}, \quad s^2_{AR} = \frac{\hat{\sigma}^2_{ek}}{(1-\hat{b}(1))^2}. \label{eq:s^2_ar}
\end{equation}
\begin{revision}
The existing literature establishes important but distinct parts of the feasible problem. In the no-deterministics case, \citet{NgPerron1996} state that the autoregressive estimator yields the same NINW limits as the oracle case under $k\to\infty$ and $k/T\to0$. \citet{NP1998} subsequently establish that $s^2_{AR} \stackrel{\IP}{\to} 0$ when the spectral density at frequency zero collapses. That is, they establish absolute consistency for the limiting value zero, but again, that by itself does not determine the ratio in \eqref{eq:rel_con}. Finally, \citet{NgPerron2001} establish that, under NINW, $k^2s^2_{AR}=O_\IP(1)$ and $Ts^2_{AR} = O_\IP(1)$ can both hold only if
\begin{equation}
    \frac{k}{\sqrt{T}} \to \kappa, \quad 0<\kappa<\infty. \label{eq:rate_con}
\end{equation}
Thus, $k \propto \sqrt{T}$ is the rate required to keep the feasible $M$-tests bounded. It remains to determine whether this rate also delivers ratio consistency. 
\end{revision}

\begin{proposition}[$s^2_{AR}$ under NINW]\label{prop:feasible_s2ar}
Let $\{y_t\}$ be defined by \eqref{eq:model}, Assumption~\ref{ass:innovations} hold, and construct $s^2_{AR}$ from GLS-detrended data as in \eqref{eq:ADF}--\eqref{eq:s^2_ar}. If the deterministic integer sequence $k=k_T$ satisfies \eqref{eq:rate_con}, then
\begin{equation}
    \frac{s^2_{AR}}{\omega^2_T} \stackrel{\IP}{\to} \coth^2\left(\frac{\kappa\delta}{2}\right), \label{eq:s^2_rel_con}
\end{equation}
equivalently, $Ts^2_{AR} \stackrel{\IP}{\to}\sigma^2_e \delta^2\coth^2(\kappa\delta/2).$ It follows for finite $\kappa$ and $\delta$ that $\coth^2(\kappa\delta/2)>1$ so the oracle and feasible $M$-test limiting distributions differ.
\end{proposition}
This proposition leads us directly to the following two theorems.

\begin{theorem}[Feasible $M^{OLS}$ limiting distributions under NINW]\label{thm:feasible_mols}
\rev{Let $\{y_t\}$ be defined by\eqref{eq:model} and Assumption~\ref{ass:innovations} hold. Use OLS-detrended data for the level moments and construct $s^2_{AR}$ from GLS-detrended data as in \eqref{eq:ADF}--\eqref{eq:s^2_ar}. If the deterministic lag sequence satisfies $k/\sqrt T\to\kappa\in(0,\infty)$, then}
\begin{align}
    \MZa^{OLS} &\wto \frac{(\eps_{\infty}+\delta J_c^\tau(1))^2 - \delta^2 \coth^2(\frac{\kappa\delta}{2})}{2 \left(1+\delta^2 \int_0^1J_c^{\tau}(r)^2dr\right)}, \notag\\
    \MSB^{OLS} &\wto \frac{\tanh(\frac{\kappa\delta}{2})}{\delta}\left(1+\delta^2 \int_0^1 J_c^{\tau}(r)^2dr\right)^{1/2} ,\label{eq:mza_msb_limt_ols_feasible} \\
    \MZt^{OLS} &\wto \frac{(\eps_{\infty}+\delta J_c^{\tau}(1))^2 - \delta^2 \coth^2(\frac{\kappa\delta}{2})}{2\delta \coth(\frac{\kappa\delta}{2})\left(1 + \delta^2 \int_0^1 J_c^{\tau}(r)^2dr\right)^{1/2}} \label{eq:mzt_limit_ols_feasible}
\end{align}

\end{theorem}
\begin{revision}
\begin{proof}
By Proposition~\ref{prop:feasible_s2ar}, $Ts^2_{AR}\stackrel{\IP}{\to}\sigma_e^2\delta^2\coth^2(\kappa\delta/2)$. Combine this with Proposition~\ref{prop:sm_ols} in the algebraic definitions of $\MZa^{OLS}$ and $\MSB^{OLS}$, and apply Slutsky's theorem. The $\MZt^{OLS}$ limit follows by multiplication.
\end{proof}
\end{revision}

\begin{theorem}[Feasible $M^{GLS}$ limiting distributions under NINW]\label{thm:feasible_mgls}
\rev{Let $\{y_t\}$ be defined by\eqref{eq:model} and Assumption~\ref{ass:innovations} hold. Use GLS-detrended data for the level moments and construct $s^2_{AR}$ from \eqref{eq:ADF}--\eqref{eq:s^2_ar}. If the deterministic lag sequence satisfies $k/\sqrt T\to\kappa\in(0,\infty)$, then}

\begin{align}
    \MZa^{GLS} &\wto \frac{[(1-\lambda)\eps_{\infty} + \frac{1-\lambda}{2}\eps_1+ \delta V_{c,\cbar}(1)]^2 - \delta^2\coth^2(\frac{\kappa\delta}{2})}{2\left(1+ \int_0^1[\delta V_{c,\cbar}(r) - \eps_1 - q_{\cbar}r]^2dr\right)}, \label{eq:mza_limit_gls_feasible} \\
    \MSB^{GLS} &\wto \frac{\tanh(\frac{\kappa\delta}{2})}{\delta}\left(1 + \int_0^1 [\delta V_{c,\cbar}(r) - \eps_1 - q_{\cbar}r]^2dr \right)^{1/2} \label{eq:msb_limit_gls_feasible} \\
    \MZt^{GLS} &\wto \frac{[(1-\lambda)\eps_{\infty} + \frac{1-\lambda}{2}\eps_1+ \delta V_{c,\cbar}(1)]^2 - \delta^2\coth^2(\frac{\kappa\delta}{2})}{2\delta\coth(\frac{\kappa\delta}{2}) \left( 1 + \int_0^1[\delta V_{c,\cbar}(r) - \eps_1 - q_{\cbar}r]^2dr\right)^{1/2}}, \label{eq:mzt_limit_gls_feasible} \\
    \MP^{GLS} &\wto \frac{\tanh^2(\frac{\kappa\delta}{2})}{\delta^2} \Bigg\{\cbar^2 \left(1 + \int_0^1[\delta V_{c,\cbar}(r) - \eps_1 - q_{\cbar}r]^2dr\right) \notag\\
    &\qquad{} + (1-\cbar) \left( (1-\lambda)\eps_{\infty} + \frac{1-\lambda}{2}\eps_1 + \delta V_{c,\cbar}(1) \right)^2 \Bigg\} \label{eq:mpt_limit_gls_feasible}
\end{align}
  
\end{theorem}
\begin{revision}
\begin{proof}
Combine Proposition~\ref{prop:feasible_s2ar} with the GLS sample-moment limits in Proposition~\ref{prop:sm_gls}. Slutsky's theorem gives the limits for $\MZa^{GLS}$ and $\MSB^{GLS}$; multiplication gives $\MZt^{GLS}$, while direct substitution into \eqref{eq:mzt_mpt} gives the $\MP^{GLS}$ limit.
\end{proof}
\end{revision}

\begin{remark}\label{rem:feasible_interpretation}
At the proportional lag rate \(k/\sqrt T\to\kappa\), the effect of
estimating the LRV is summarized by $\coth^2\left(\frac{\kappa\delta}{2}\right).$
Terms involving the LRV itself acquire this inflation
factor, while terms involving its inverse square root acquire the factor
$\tanh\left(\frac{\kappa\delta}{2}\right).$

For finite \(\kappa\), these factors differ from one, so the feasible
limits do not coincide with the oracle limits. They approach the oracle
limits only as \(\kappa\delta\) becomes large. Hence, when the selected lag order is sufficiently large relative to the decay rate of the moving average component, the feasible tests begin to inherit the oracle behavior. In particular, for values of $\delta$ that lie in the zero size region identified in Corollary \ref{cor:support_bounds}, conventional critical values can lead to under-rejection. It is important to note, however, that the feasible zero asymptotic rejection probability region, when it exists, will always be contained within the corresponding oracle region as a result of LRV overestimation counteracting the oracle under-rejection.
\end{remark}

\begin{remark}
\label{rem:kappadelta}
The coefficients in the autoregressive representation of the
near-canceling moving average component decay approximately as
\[
\left(1-\frac{\delta}{\sqrt T}\right)^j
\approx
\exp\left(-\frac{\delta j}{\sqrt T}\right).
\]
When $k/\sqrt T\to\kappa$, the last included coefficient is therefore
of approximate order $e^{-\kappa\delta}$. The product
$\kappa\delta$ measures how far the lag truncation extends into this
decaying sequence, which explains why the additional distortion caused
by LRV estimation depends on $\kappa$ and $\delta$
through their product.
\end{remark}

\section{Finite sample simulations}\label{sec:finite_sample}
The preceding oracle and feasible results suggest two sources of nonstandard behavior under NINW: (1) when the LRV is known, the $M$-tests undersize under conventional critical values for sufficiently small $\delta$, and (2) at the proportional lag rate, $s^2_{AR}$ asymptotically overestimates the LRV, increasing rejection. Hence, the finite sample analysis explores how these two factors interact under data-dependent lag selection rules.

Proposition \ref{prop:feasible_s2ar} and Theorems \ref{thm:feasible_mols}--\ref{thm:feasible_mgls} concern deterministic lag sequences satisfying $k/\sqrt T\to\kappa\in(0,\infty)$. The simulations examine how the mechanisms identified by these results manifest under data-dependent lag selection by MAIC and MBIC. Establishing limiting distributions for these selected-lag procedures is beyond the scope of this paper.

\subsection{Lag selection rules}
In this paper, we consider only the modified criteria $MAIC^{GLS}, MBIC^{GLS}, MAIC^{OLS},$ and $MBIC^{OLS}$ which originate from \citet{NgPerron2001} and \citet{PQ2007}, respectively. We do not consider the standard criteria like AIC and BIC as can be seen in \citet{akaike_new_1974} and \citet{BIC1974}, respectively, as \citet{NgPerron2001} already show their unsuitability under DGPs featuring large negative moving average coefficients. 

The modified criteria choose a value $k_{MIC} = \arg\min_{k\in\{0,\ldots,k_{max}\}} MIC(k)$ where
\begin{equation}
    MIC(k) = \ln(\hat{\sigma}^2_k) + \frac{C_T(\tau_T(k)+k)}{T-k_{max}-1} \label{eq:mic},
\end{equation}
and
\begin{equation}
    \hat{\sigma}^2_k = \frac{1}{T-k_{max}-1}\sum_{t=k_{max}+2}^T \hat{e}^2_{t,k}, \quad \tau_T(k) = \frac{1}{\hat{\sigma}^2_k}\hat{b}^2_0\sum_{t=k_{max}+2}^T \ty^2_{t-1}. \label{eq:mic_components}
\end{equation}
$k_{max}$ is usually set to $k_{max}=[ 12(T/100)^{1/4}]$ in accordance with \citet{Schwert1989} and for MAIC $C_T=2$ whereas for MBIC $C_T=\ln(T-k_{max}-1)$. For the $MIC^{GLS}$ criteria, \eqref{eq:mic} is fed GLS detrended data while, for the $MIC^{OLS}$ criteria, it is fed OLS detrended data. 

\subsection{Size results}\label{sec:size}
We hold $\delta$ fixed across sample sizes and $e_0=u_0=0$ with $e_t\stackrel{\mathrm{i.i.d.}}{\sim}N(0,1)$.Unless otherwise noted, 30,000 Monte Carlo replications are used. $s^2_{AR}$ as defined in \eqref{eq:s^2_ar} is constructed from GLS detrended data for all procedures. To allow lag orders at the rate required by Proposition \ref{prop:feasible_s2ar}, we set $k_{max}=[ 2\sqrt{T} ] $, which yields $k_{max} = 20,31,44$, for $T=100,250, 500$, respectively\footnote{This does not necessarily imply that this extended cap is preferable to the conventional cap. But, under the conventional cap, oversizing is more extreme at low delta and exhibits similar undersizing elsewhere.}. This choice differs from that of \citet{NgPerron2001} and \citet{PQ2007} as we find, for large negative moving average coefficients, $k$ will often bind to $k_{max}$ when the traditional cap is used.

\begin{sidewaystable}[p]
\centering
\normalsize\singlespacing
\begin{threeparttable}
\caption{Empirical size of the $MZ_\alpha^{OLS}$ and $MZ_\alpha^{GLS}$ tests using conventional 5\% asymptotic critical values under extended lag cap}
\label{tab:mza-delta-asymptotic-size}
\setlength{\tabcolsep}{5.5pt}
\renewcommand{\arraystretch}{0.8}
\begin{tabular}{c S[table-format=1.1] S[table-format=-1.3] *{8}{S[table-format=1.3]}}
\toprule
& &
& \multicolumn{4}{c}{$MZ_\alpha^{OLS}$}
& \multicolumn{4}{c}{$MZ_\alpha^{GLS}$} \\
\cmidrule(lr){4-7}\cmidrule(lr){8-11}
{$T$} & {$\delta$} & {$\theta$}
& {$MAIC^{GLS}$}
& {$MAIC^{OLS}$}
& {$MBIC^{GLS}$}
& {$MBIC^{OLS}$}
& {$MAIC^{GLS}$}
& {$MAIC^{OLS}$}
& {$MBIC^{GLS}$}
& {$MBIC^{OLS}$}
\\
\midrule
\multirow{10}{*}{100}
& 0.5 & -0.950 & 0.183 & 0.485 & 0.198 & 0.503 & 0.194 & 0.499 & 0.210 & 0.516 \\
& 1.0 & -0.900 & 0.124 & 0.297 & 0.135 & 0.313 & 0.137 & 0.312 & 0.150 & 0.328 \\
& 1.5 & -0.850 & 0.075 & 0.166 & 0.083 & 0.177 & 0.086 & 0.179 & 0.097 & 0.192 \\
& 2.0 & -0.800 & 0.053 & 0.103 & 0.060 & 0.111 & 0.065 & 0.117 & 0.074 & 0.127 \\
& 2.5 & -0.750 & 0.037 & 0.066 & 0.042 & 0.073 & 0.048 & 0.078 & 0.055 & 0.085 \\
& 3.0 & -0.700 & 0.031 & 0.050 & 0.035 & 0.056 & 0.042 & 0.062 & 0.049 & 0.069 \\
& 3.5 & -0.650 & 0.026 & 0.042 & 0.030 & 0.046 & 0.038 & 0.054 & 0.044 & 0.059 \\
& 4.0 & -0.600 & 0.026 & 0.037 & 0.030 & 0.041 & 0.039 & 0.049 & 0.045 & 0.054 \\
& 4.5 & -0.550 & 0.026 & 0.036 & 0.030 & 0.040 & 0.040 & 0.050 & 0.046 & 0.054 \\
& 5.0 & -0.500 & 0.026 & 0.034 & 0.030 & 0.038 & 0.039 & 0.047 & 0.045 & 0.052 \\
\specialrule{1.05pt}{2pt}{2pt}
\multirow{10}{*}{250}
& 0.5 & -0.968 & 0.059 & 0.353 & 0.064 & 0.369 & 0.063 & 0.356 & 0.069 & 0.372 \\
& 1.0 & -0.937 & 0.027 & 0.131 & 0.030 & 0.140 & 0.030 & 0.135 & 0.034 & 0.144 \\
& 1.5 & -0.905 & 0.013 & 0.048 & 0.015 & 0.053 & 0.016 & 0.052 & 0.018 & 0.057 \\
& 2.0 & -0.874 & 0.009 & 0.024 & 0.010 & 0.027 & 0.011 & 0.027 & 0.014 & 0.030 \\
& 2.5 & -0.842 & 0.006 & 0.014 & 0.007 & 0.016 & 0.008 & 0.017 & 0.011 & 0.020 \\
& 3.0 & -0.810 & 0.006 & 0.012 & 0.008 & 0.014 & 0.010 & 0.016 & 0.013 & 0.019 \\
& 3.5 & -0.779 & 0.007 & 0.012 & 0.010 & 0.014 & 0.012 & 0.017 & 0.016 & 0.020 \\
& 4.0 & -0.747 & 0.007 & 0.011 & 0.010 & 0.014 & 0.013 & 0.018 & 0.018 & 0.022 \\
& 4.5 & -0.715 & 0.010 & 0.013 & 0.014 & 0.016 & 0.017 & 0.020 & 0.023 & 0.025 \\
& 5.0 & -0.684 & 0.012 & 0.015 & 0.017 & 0.019 & 0.019 & 0.022 & 0.025 & 0.027 \\
\specialrule{1.05pt}{2pt}{2pt}
\multirow{10}{*}{500}
& 0.5 & -0.978 & 0.022 & 0.215 & 0.024 & 0.227 & 0.024 & 0.217 & 0.026 & 0.229 \\
& 1.0 & -0.955 & 0.006 & 0.041 & 0.007 & 0.045 & 0.007 & 0.043 & 0.008 & 0.047 \\
& 1.5 & -0.933 & 0.002 & 0.010 & 0.002 & 0.011 & 0.003 & 0.012 & 0.003 & 0.013 \\
& 2.0 & -0.911 & 0.002 & 0.004 & 0.002 & 0.005 & 0.002 & 0.006 & 0.003 & 0.006 \\
& 2.5 & -0.888 & 0.001 & 0.002 & 0.001 & 0.003 & 0.002 & 0.004 & 0.002 & 0.004 \\
& 3.0 & -0.866 & 0.001 & 0.002 & 0.002 & 0.003 & 0.002 & 0.004 & 0.003 & 0.005 \\
& 3.5 & -0.843 & 0.002 & 0.003 & 0.002 & 0.003 & 0.003 & 0.004 & 0.004 & 0.005 \\
& 4.0 & -0.821 & 0.002 & 0.003 & 0.004 & 0.004 & 0.004 & 0.005 & 0.006 & 0.007 \\
& 4.5 & -0.799 & 0.003 & 0.004 & 0.006 & 0.005 & 0.006 & 0.007 & 0.010 & 0.010 \\
& 5.0 & -0.776 & 0.004 & 0.005 & 0.009 & 0.007 & 0.008 & 0.009 & 0.013 & 0.013 \\
\bottomrule
\end{tabular}
\begin{tablenotes}[flushleft]\footnotesize\singlespacing
\item \textit{Notes:} Critical values are -20.5 for $\MZa^{OLS}$ and -17.3 for $\MZa^{GLS}.$
\end{tablenotes}
\end{threeparttable}
\end{sidewaystable}

 Table \ref{tab:mza-delta-asymptotic-size} shows empirical size results for $\MZa$ for brevity as all feasible tests display qualitatively similar results. The main pattern as $T$ increases is a transition from oversizing to undersizing at each fixed $\delta$, or undersizing to more extreme undersizing. Hence, a certain cell reaching near nominal size is more coincidental to the exact parameter specification rather than a well-calibrated test. In addition, $\MZa^{OLS}$ is generally more undersized (or less oversized) than $\MZa^{GLS}$ across all procedures. Also, the $MIC^{OLS}$ criteria oversize relative to their GLS counterparts, which is especially strong at small $\delta$ as a result of selecting significantly smaller $k$. Table~\ref{tab:mza-maicgls-summary-sqrt2T} provides a more focused explanation of these results for the $MAIC^{GLS}$ rule applied to $\MZa^{GLS}$ (although similar results hold for all specifications).  

\begin{table}[t]
\centering
\normalsize\singlespacing
\begin{threeparttable}
\caption{Size breakdown for $\MZa^{GLS}$ with $MAIC^{GLS}$ under extended lag cap}
\label{tab:mza-maicgls-summary-sqrt2T}
\setlength{\tabcolsep}{8pt}
\renewcommand{\arraystretch}{1.15}
\begin{tabular}{c S[table-format=-1.3] S[table-format=1.1] S[table-format=1.3] S[table-format=1.3] S[table-format=1.3] S[table-format=2.1] S[table-format=1.3] S[table-format=2.2] S[table-format=1.3]}
\toprule
& & & \multicolumn{3}{c}{$MZ_\alpha^{GLS}$ size} & & & & \\
\cmidrule(lr){4-6}
{$T$} & {$\theta$} & {$\delta$} & {feasible} & {$R^{pop}_k\omega^2$} & {oracle} & {$k$} & {$k/\sqrt{T}$} & {$ s^2_{AR}/\omega^2$} & {$ H$} \\
\midrule
\multirow{4}{*}{100} & -0.950 & 0.5 & 0.194 & 0.129 & 0.000 & 10.0 & 1.000 & 16.80 & 1.152 \\
& -0.850 & 1.5 & 0.086 & 0.043 & 0.000 & 8.0 & 0.800 & 2.98 & 1.072 \\
& -0.750 & 2.5 & 0.048 & 0.018 & 0.000 & 6.0 & 0.600 & 1.93 & 1.001 \\
& -0.500 & 5.0 & 0.039 & 0.026 & 0.001 & 3.0 & 0.300 & 1.28 & 0.928 \\
\midrule
\multirow{4}{*}{250} & -0.968 & 0.5 & 0.063 & 0.032 & 0.000 & 21.0 & 1.328 & 11.29 & 1.187 \\
& -0.905 & 1.5 & 0.016 & 0.004 & 0.000 & 15.0 & 0.949 & 2.41 & 1.035 \\
& -0.842 & 2.5 & 0.008 & 0.001 & 0.000 & 11.0 & 0.696 & 1.77 & 0.996 \\
& -0.684 & 5.0 & 0.019 & 0.004 & 0.000 & 6.0 & 0.379 & 1.34 & 0.961 \\
\midrule
\multirow{4}{*}{500} & -0.978 & 0.5 & 0.024 & 0.010 & 0.000 & 32.0 & 1.431 & 10.25 & 1.134 \\
& -0.933 & 1.5 & 0.003 & 0.000 & 0.000 & 23.0 & 1.029 & 2.36 & 1.017 \\
& -0.888 & 2.5 & 0.002 & 0.000 & 0.000 & 16.0 & 0.716 & 1.77 & 0.988 \\
& -0.776 & 5.0 & 0.008 & 0.001 & 0.000 & 9.0 & 0.402 & 1.37 & 0.970 \\
\bottomrule
\end{tabular}
\begin{tablenotes}[flushleft]\footnotesize\singlespacing
\item \textit{Notes}: The feasible, \(R_k^{pop}\omega_T^2\), and oracle columns report rejection frequencies using \(s_{AR}^2\), the population AR(\(k\)) LRV evaluated at each replication’s selected lag, and the true LRV, respectively, in the same statistic with critical value \(-17.3\). The middle column retains population truncation distortion while removing LRV coefficient-estimation error. The columns \(k\), \(k/\sqrt T\), \(s_{AR}^2/\omega_T^2\), and \(H\) report medians across replications.
\end{tablenotes}
\end{threeparttable}
\end{table}

 Two important definitions are 
 \begin{equation}
     H = \frac{s^2_{AR}/\omega_T^2}{R^{pop}_k}, \qquad R^{pop}_k = \frac{s^{2,pop}_{AR}}{\omega_T^2} = \frac{(1+(-\theta_T)^{k+1})(1+(-\theta_T)^{k+2})}{(1-(-\theta_T)^{k+1})(1-(-\theta_T)^{k+2})}, \label{eq:H}
 \end{equation}
 so that $R^{pop}_k$ is the inflation of the LRV estimate as a result of truncating the infinite AR process at lag $k$ even if the population AR coefficients are known\footnote{The expression for $R^{pop}_k$ follows from solving the population Yule--Walker equations for the AR($k$) best linear approximation to the stationary MA(1) process and evaluating its implied LRV. Note that under the lag rate in \ref{eq:rate_con}, $R^{pop}_k \to \coth^2(\kappa\delta/2).$}. Hence, $H$ (calculated per replication) measures the remaining finite sample distortion in $s^2_{AR}/\omega_T^2$ after accounting for this population truncation effect. Hence, $H=1$ implies that the realized LRV inflation is exactly that implied by the truncation at the selected lag. The oracle rejection rate is essentially zero throughout the table, confirming the strong conservatism of the statistic under conventional critical values. Population truncation offsets part of this conservatism. For example, at $T=100, \delta=0.5$, replacing the true LRV by $R^{pop}_k \omega_T^2$ raises rejection from near zero to 0.129, compared with feasible size of 0.194. The remaining difference between the latter two can be explained by estimation error beyond the population truncation benchmark. Similarly, the near nominal feasible size at $T=100, \delta=2.5$ combines an oracle size of near zero and a truncation-only size of 0.018, and an additional finite sample estimation effect. Thus, its proximity to 0.05 does not indicate a well-calibrated test.

 Median $s^2_{AR}/\omega_T^2$ is consistently over one and is particularly large at smaller $\delta$. Median $H$ is close to one, indicating that population AR truncation explains much of the LRV overestimation. This alone does not imply that truncation alone explains feasible size since rejection depends on the upper tail of $s^2_{AR}$, whereas $H$ reports a median. Indeed, the consistently higher feasible than truncation-only sizes illustrate the residual estimation variability is consequential.

 As $T$ increases, feasible size declines toward the conservative oracle behavior. For small and moderate $\delta$, this is accompanied by a slowing decline in LRV overestimation (at $\delta=5$ the overestimation ratio modestly rises). Overall, the results support that feasible size is a result of an interaction between intrinsic oracle conservatism and truncation induced LRV overestimation. 

    The previous simulations hold $\delta$ fixed as $T$ increases, differing from the conventional practice of holding $\theta$ fixed. We repeated the size simulations, holding $\theta$ fixed at -0.9 and -0.8, for which $\delta_T = 0.1\sqrt{T}$ and $\delta_T = 0.2\sqrt{T}$, respectively. In this case, as $T$ increases, so does $\delta$ at each fixed $\theta$, so the DGP moves further away from the small $\delta$ NINW region as $T$ increases. We find that size exhibits a non-monotone pattern in this case: at $T=100$, the tests are oversized, and at $T=250,500,1000$, size is near 0, while by $T=10000$, size begins to increase towards the nominal level. This is consistent with the idea that at small sample sizes LRV overestimation dominates the intrinsic conservatism displayed by the statistics. In addition, large sample simulations were run with $\delta=0.5$ and $\delta=5$ for up to $T=20000$ and we find that feasible size decreases towards zero for all lag selection rules (albeit at different rates for each rule), consistent with Remark \ref{rem:feasible_interpretation}.

\subsection{Size-adjusted power results}\label{sec:power}
To assess the power of the tests, we set $c=-13.5$ for the local alternative. We now consider more specifically the trade-offs of using the extended cap or the conventional cap for $k$.

\begin{sidewaystable}[tp]
\centering
\normalsize\singlespacing
\begin{threeparttable}
\caption{Size-adjusted power of the $MZ_\alpha^{OLS}$ and $MZ_\alpha^{GLS}$ tests under extended lag cap}
\label{tab:mza-delta-sa-power-extended}
\setlength{\tabcolsep}{5.5pt}
\renewcommand{\arraystretch}{0.8}
\begin{tabular}{c S[table-format=1.1] S[table-format=-1.3] *{8}{S[table-format=1.3]}}
\toprule
& &
& \multicolumn{4}{c}{$MZ_\alpha^{OLS}$}
& \multicolumn{4}{c}{$MZ_\alpha^{GLS}$} \\
\cmidrule(lr){4-7}\cmidrule(lr){8-11}
{$T$} & {$\delta$} & {$\theta$}
& {$MAIC^{GLS}$}
& {$MAIC^{OLS}$}
& {$MBIC^{GLS}$}
& {$MBIC^{OLS}$}
& {$MAIC^{GLS}$}
& {$MAIC^{OLS}$}
& {$MBIC^{GLS}$}
& {$MBIC^{OLS}$}
\\
\midrule
\multirow{10}{*}{100}
& 0.5 & -0.950 & 0.105 & 0.123 & 0.104 & 0.119 & 0.112 & 0.111 & 0.111 & 0.110 \\
& 1.0 & -0.900 & 0.151 & 0.185 & 0.147 & 0.182 & 0.153 & 0.202 & 0.152 & 0.202 \\
& 1.5 & -0.850 & 0.174 & 0.259 & 0.179 & 0.254 & 0.174 & 0.279 & 0.179 & 0.281 \\
& 2.0 & -0.800 & 0.165 & 0.273 & 0.171 & 0.277 & 0.165 & 0.273 & 0.171 & 0.275 \\
& 2.5 & -0.750 & 0.171 & 0.254 & 0.174 & 0.256 & 0.172 & 0.254 & 0.175 & 0.255 \\
& 3.0 & -0.700 & 0.177 & 0.236 & 0.182 & 0.237 & 0.178 & 0.237 & 0.184 & 0.238 \\
& 3.5 & -0.650 & 0.178 & 0.226 & 0.183 & 0.229 & 0.181 & 0.228 & 0.186 & 0.233 \\
& 4.0 & -0.600 & 0.199 & 0.228 & 0.202 & 0.231 & 0.200 & 0.235 & 0.207 & 0.237 \\
& 4.5 & -0.550 & 0.207 & 0.231 & 0.212 & 0.235 & 0.212 & 0.235 & 0.221 & 0.240 \\
& 5.0 & -0.500 & 0.222 & 0.247 & 0.235 & 0.252 & 0.227 & 0.253 & 0.240 & 0.260 \\
\specialrule{1.05pt}{2pt}{2pt}
\multirow{10}{*}{250}
& 0.5 & -0.968 & 0.112 & 0.156 & 0.114 & 0.151 & 0.112 & 0.178 & 0.114 & 0.175 \\
& 1.0 & -0.937 & 0.110 & 0.270 & 0.112 & 0.269 & 0.110 & 0.265 & 0.112 & 0.264 \\
& 1.5 & -0.905 & 0.103 & 0.253 & 0.103 & 0.253 & 0.104 & 0.252 & 0.106 & 0.251 \\
& 2.0 & -0.874 & 0.112 & 0.218 & 0.109 & 0.220 & 0.113 & 0.214 & 0.115 & 0.215 \\
& 2.5 & -0.842 & 0.107 & 0.185 & 0.106 & 0.186 & 0.112 & 0.183 & 0.113 & 0.184 \\
& 3.0 & -0.810 & 0.124 & 0.184 & 0.127 & 0.183 & 0.132 & 0.183 & 0.136 & 0.181 \\
& 3.5 & -0.779 & 0.141 & 0.187 & 0.140 & 0.185 & 0.144 & 0.185 & 0.150 & 0.183 \\
& 4.0 & -0.747 & 0.162 & 0.198 & 0.166 & 0.203 & 0.168 & 0.199 & 0.178 & 0.204 \\
& 4.5 & -0.715 & 0.182 & 0.215 & 0.187 & 0.221 & 0.186 & 0.215 & 0.197 & 0.221 \\
& 5.0 & -0.684 & 0.197 & 0.224 & 0.198 & 0.227 & 0.207 & 0.229 & 0.216 & 0.233 \\
\specialrule{1.05pt}{2pt}{2pt}
\multirow{10}{*}{500}
& 0.5 & -0.978 & 0.088 & 0.209 & 0.088 & 0.206 & 0.088 & 0.217 & 0.088 & 0.215 \\
& 1.0 & -0.955 & 0.086 & 0.265 & 0.083 & 0.266 & 0.091 & 0.261 & 0.086 & 0.262 \\
& 1.5 & -0.933 & 0.080 & 0.196 & 0.067 & 0.195 & 0.086 & 0.191 & 0.078 & 0.190 \\
& 2.0 & -0.911 & 0.078 & 0.156 & 0.066 & 0.152 & 0.087 & 0.152 & 0.079 & 0.146 \\
& 2.5 & -0.888 & 0.094 & 0.151 & 0.081 & 0.145 & 0.102 & 0.149 & 0.098 & 0.142 \\
& 3.0 & -0.866 & 0.109 & 0.154 & 0.093 & 0.150 & 0.116 & 0.150 & 0.112 & 0.145 \\
& 3.5 & -0.843 & 0.134 & 0.172 & 0.119 & 0.168 & 0.144 & 0.172 & 0.141 & 0.165 \\
& 4.0 & -0.821 & 0.151 & 0.176 & 0.138 & 0.180 & 0.159 & 0.177 & 0.160 & 0.177 \\
& 4.5 & -0.799 & 0.177 & 0.199 & 0.160 & 0.202 & 0.187 & 0.201 & 0.188 & 0.203 \\
& 5.0 & -0.776 & 0.201 & 0.221 & 0.191 & 0.221 & 0.212 & 0.225 & 0.214 & 0.224 \\
\bottomrule
\end{tabular}
\begin{tablenotes}[flushleft]\footnotesize\singlespacing
\item \textit{Notes:} Size-adjusted power against the local alternative $\rho=1+c/T$ with $c=-13.5$ ($\rho=0.865, 0.946, 0.973$ for $T=100,\,250,\,500$, respectively). Lags are selected from $k\in[0,k_{max}]$ with $k_{max}=20, 31, 44$ for $T=100,\,250,\,500$, respectively.
\end{tablenotes}
\end{threeparttable}
\end{sidewaystable}
\begin{sidewaystable}[tp]
\centering
\normalsize\singlespacing
\begin{threeparttable}
\caption{Size-adjusted power of the $MZ_\alpha^{OLS}$ and $MZ_\alpha^{GLS}$ tests under conventional lag cap}
\label{tab:mza-delta-sa-power-conventional}
\setlength{\tabcolsep}{5.5pt}
\renewcommand{\arraystretch}{0.8}
\begin{tabular}{c S[table-format=1.1] S[table-format=-1.3] *{8}{S[table-format=1.3]}}
\toprule
& &
& \multicolumn{4}{c}{$MZ_\alpha^{OLS}$}
& \multicolumn{4}{c}{$MZ_\alpha^{GLS}$} \\
\cmidrule(lr){4-7}\cmidrule(lr){8-11}
{$T$} & {$\delta$} & {$\theta$}
& {$MAIC^{GLS}$}
& {$MAIC^{OLS}$}
& {$MBIC^{GLS}$}
& {$MBIC^{OLS}$}
& {$MAIC^{GLS}$}
& {$MAIC^{OLS}$}
& {$MBIC^{GLS}$}
& {$MBIC^{OLS}$}
\\
\midrule
\multirow{10}{*}{100}
& 0.5 & -0.950 & 0.120 & 0.115 & 0.119 & 0.113 & 0.119 & 0.107 & 0.118 & 0.105 \\
& 1.0 & -0.900 & 0.167 & 0.179 & 0.167 & 0.176 & 0.167 & 0.198 & 0.167 & 0.194 \\
& 1.5 & -0.850 & 0.182 & 0.252 & 0.184 & 0.248 & 0.182 & 0.277 & 0.185 & 0.274 \\
& 2.0 & -0.800 & 0.175 & 0.280 & 0.176 & 0.280 & 0.175 & 0.278 & 0.175 & 0.280 \\
& 2.5 & -0.750 & 0.175 & 0.262 & 0.179 & 0.264 & 0.175 & 0.262 & 0.179 & 0.263 \\
& 3.0 & -0.700 & 0.179 & 0.240 & 0.181 & 0.239 & 0.180 & 0.240 & 0.183 & 0.240 \\
& 3.5 & -0.650 & 0.184 & 0.232 & 0.186 & 0.234 & 0.188 & 0.234 & 0.189 & 0.238 \\
& 4.0 & -0.600 & 0.200 & 0.234 & 0.200 & 0.235 & 0.200 & 0.239 & 0.203 & 0.241 \\
& 4.5 & -0.550 & 0.212 & 0.237 & 0.213 & 0.240 & 0.216 & 0.244 & 0.219 & 0.247 \\
& 5.0 & -0.500 & 0.228 & 0.254 & 0.232 & 0.256 & 0.228 & 0.258 & 0.236 & 0.262 \\
\specialrule{1.05pt}{2pt}{2pt}
\multirow{10}{*}{250}
& 0.5 & -0.968 & 0.131 & 0.138 & 0.130 & 0.136 & 0.130 & 0.161 & 0.130 & 0.160 \\
& 1.0 & -0.937 & 0.145 & 0.259 & 0.145 & 0.254 & 0.144 & 0.255 & 0.145 & 0.252 \\
& 1.5 & -0.905 & 0.144 & 0.268 & 0.142 & 0.271 & 0.144 & 0.267 & 0.142 & 0.268 \\
& 2.0 & -0.874 & 0.161 & 0.240 & 0.154 & 0.240 & 0.161 & 0.236 & 0.159 & 0.237 \\
& 2.5 & -0.842 & 0.151 & 0.208 & 0.145 & 0.211 & 0.156 & 0.208 & 0.150 & 0.206 \\
& 3.0 & -0.810 & 0.161 & 0.205 & 0.151 & 0.200 & 0.167 & 0.204 & 0.161 & 0.199 \\
& 3.5 & -0.779 & 0.165 & 0.200 & 0.156 & 0.198 & 0.169 & 0.200 & 0.166 & 0.194 \\
& 4.0 & -0.747 & 0.186 & 0.218 & 0.177 & 0.213 & 0.193 & 0.220 & 0.191 & 0.216 \\
& 4.5 & -0.715 & 0.200 & 0.227 & 0.198 & 0.232 & 0.203 & 0.226 & 0.207 & 0.229 \\
& 5.0 & -0.684 & 0.211 & 0.234 & 0.209 & 0.237 & 0.220 & 0.239 & 0.224 & 0.241 \\
\specialrule{1.05pt}{2pt}{2pt}
\multirow{10}{*}{500}
& 0.5 & -0.978 & 0.117 & 0.171 & 0.118 & 0.168 & 0.117 & 0.186 & 0.118 & 0.184 \\
& 1.0 & -0.955 & 0.144 & 0.264 & 0.143 & 0.265 & 0.143 & 0.260 & 0.143 & 0.262 \\
& 1.5 & -0.933 & 0.166 & 0.227 & 0.163 & 0.229 & 0.165 & 0.220 & 0.165 & 0.222 \\
& 2.0 & -0.911 & 0.181 & 0.212 & 0.160 & 0.210 & 0.181 & 0.209 & 0.180 & 0.208 \\
& 2.5 & -0.888 & 0.186 & 0.208 & 0.155 & 0.203 & 0.188 & 0.204 & 0.177 & 0.200 \\
& 3.0 & -0.866 & 0.185 & 0.206 & 0.145 & 0.196 & 0.186 & 0.201 & 0.173 & 0.191 \\
& 3.5 & -0.843 & 0.196 & 0.218 & 0.159 & 0.207 & 0.204 & 0.217 & 0.184 & 0.201 \\
& 4.0 & -0.821 & 0.198 & 0.218 & 0.166 & 0.206 & 0.202 & 0.215 & 0.189 & 0.205 \\
& 4.5 & -0.799 & 0.210 & 0.230 & 0.182 & 0.218 & 0.225 & 0.234 & 0.212 & 0.222 \\
& 5.0 & -0.776 & 0.231 & 0.244 & 0.204 & 0.234 & 0.238 & 0.246 & 0.227 & 0.238 \\
\bottomrule
\end{tabular}
\begin{tablenotes}[flushleft]\footnotesize\singlespacing
\item \textit{Notes:} Size-adjusted power against the local alternative $\rho=1+c/T$ with $c=-13.5$ ($\rho=0.865, 0.946, 0.973$ for $T=100,\,250,\,500$, respectively). Lags are selected from $k\in[0,k_{max}]$ with $k_{max}=12, 15, 17$ for $T=100,\,250,\,500$, respectively.
\end{tablenotes}
\end{threeparttable}
\end{sidewaystable}

Tables~\ref{tab:mza-delta-sa-power-extended} and \ref{tab:mza-delta-sa-power-conventional} report size-adjusted power under the extended cap and conventional cap for $\MZa$, respectively\footnote{Note that $\MSB^{OLS}$ produced similar power results to that of $\MZa$ in spite of its favorable oracle power function, demonstrating that oracle properties need not pass on to the feasible tests.}. A clear result of the extended cap is that power frequently deteriorates as the sample size increases. For $\delta=1.5,\dots,4.5$, power at $T=500$ is below its $T=100$ value for every procedure, and all but two of these procedures have power decrease monotonically. Since each cell is size-adjusted by construction, this decline reflects weaker discriminatory ability between the null and local alternative, not increasing conservatism under conventional critical values. On the other hand, the conventional cap generally permits greater size-adjusted power but weakly and not uniformly. In addition, the systematic decline in power as the sample size increases is less systematic: many procedures lose power from $T=100$ to $T=250$, but recover by $T=500$, although continued declines and other patterns are also present.

For a fixed lag selection rule, $\MZa^{GLS}$ has a slight power advantage over $\MZa^{OLS}$, but this ranking is not uniformly true. By contrast, OLS lag selection rules generally produce greater size-adjusted power than their GLS counterparts under both caps. This advantage is particularly apparent under the extended cap, but remains pertinent under the conventional cap as well. Differences between MAIC and MBIC are usually secondary to those associated with the cap and GLS or OLS construction of the information criterion though not entirely negligible. 

A related pattern appears in Table~VI.B of \citet{NgPerron2001}. For $\MZa^{GLS}$ and $MAIC^{GLS}$, at $\theta=-0.8$, the table shows size-adjusted power declining as $T$ increases. Their simulations were fixed-$\theta$ while ours are fixed-$\delta$, so they are not directly comparable, but their results show that decreasing size-adjusted power with sample size is not unique to the simulation design present here.

We also ran simulations which replaced the Gaussian innovation distribution by standardized $t_5$ and Rademacher innovations. Rademacher innovations resulted in a substantial decrease in size-adjusted power while $t_5$ innovations slightly increased size-adjusted power in general. These results support the importance of the innovation distribution dependence observed in Remark \ref{rem:innovation_dep}.

Another complicating issue not shown in these simulations, is the problem of power reversal. \citet{PQ2007} recommend the OLS construction of the modified information criterion instead of GLS as a remedy, and our results show it to generally be favorable in terms of size adjusted power. However, the OLS criteria drastically oversize at low $\delta$ and severely undersize elsewhere, producing a severe size-power tradeoff.

\FloatBarrier

\section{Conclusion}\label{sec:conclusion}
We illustrate the limitations of the $M$-tests under nearly integrated nearly white noise and an unknown linear time trend. The tests' limiting distributions are contaminated by noise terms, with the $M^{GLS}$ tests being more so, causing dependence on the innovation distribution. When the LRV is known, conventional critical values result in a region of zero asymptotic rejection probability for sufficiently small $\delta$ and several tests exhibit substantial loss of asymptotic power. Yet, the Gaussian power envelope itself remains identical to the standard envelope in \citet{ERS1996}. Hence, power losses reflect inefficiency of the tests in utilizing the information available rather than a deterioration of the information itself. The oracle $\MSB^{OLS}$ test's power function provides an exception.

Feasible implementation introduces a separate distortion. For deterministic lag orders satisfying $k/\sqrt{T} \to \kappa$ $(0 < \kappa < \infty)$, the autoregressive LRV estimator converges, relative to the true LRV, to $\coth^2(\kappa\delta/2)$. Thus, oracle and feasible limiting distribution equivalence does not hold under the lag rate which keeps the $M$-tests bounded. The feasible tests face two forces which pull in different directions: LRV overestimation which encourages rejection and oracle undersizing. This implies, near nominal size, when obtained, is a result of these two forces balancing rather than a well calibrated test.

The finite sample simulations under data-dependent lag selection rules shows the practical consequences of this interaction. The choice of lag cap and detrending for the modified information criterion result in vast differences in size and size-adjusted power, and limitations exist for all that are not resolved as the sample size increases. Among the procedures examined, no single one is uniformly satisfactory. These results motivate further research into procedures designed explicitly to handle a vanishing LRV.

\newpage
\appendix
\numberwithin{equation}{section}
\renewcommand{\theHequation}{appendix.\thesection.\arabic{equation}}
\renewcommand{\theHappendixlemma}{appendix.\thesection.\arabic{appendixlemma}}
\let\lemma\appendixlemma
\let\endlemma\endappendixlemma
\makeatletter
\renewcommand{\@seccntformat}[1]{}
\makeatother
\section*{\LARGE Appendices}
Appendix \ref{sec:appendix_a} proves Propositions \ref{prop:sm_ols}--\ref{prop:sm_gls}. Appendix \ref{sec:appendix_b} proves Theorem \ref{thm:power_envelope}. Appendix \ref{sec:appendix_c} proves Proposition \ref{prop:feasible_s2ar}
\section{Appendix A}\label{sec:appendix_a}
\begin{revision}
This appendix works under model~\eqref{eq:model} and Assumption~\ref{ass:innovations}. Define
\[
W_T(r) = \frac{1}{\sigma_e \sqrt{T}}\sum_{t=1}^{[Tr]} e_t, \quad 0 \leq r \leq 1.
\]
Throughout, $\wto$ denotes convergence in distribution and $\stackrel{\IP}{\to}$ denotes convergence in probability. Process convergence is in $D[0,1]$; joint convergence with scalar or vector quantities is in the corresponding product space.
\end{revision}
\begin{lemma}\label{A.0}
    \rev{Grant Assumption~\ref{ass:innovations} and let $F$ denote the distribution of $e_t/\sigma_e$.} Then there exist independent random variables $\varepsilon_1, \varepsilon_\infty \sim F$ and a standard Brownian motion $W$, independent of $(\varepsilon_1, \varepsilon_\infty)$, such that
    \[
    \left(\frac{e_1}{\sigma_e}, \frac{e_T}{\sigma_e} , W_T\right) \wto (\varepsilon_1, \varepsilon_{\infty}, W)
    \]
    in $\mathbb{R}^2 \times D[0,1].$
\end{lemma}

\begin{proof}
    Define the trimmed partial sum process 
    \[
    W^*_T(r) = \frac{1}{\sigma_e \sqrt{T}} \sum_{t=2}^{\min([Tr], T-1)} e_t.
    \]
    For each $T$, $W^*_T$ is independent of $e_1$ and $e_T$. By the FCLT, $W^*_T \wto W$. Moreover,
    \[
    \sup_{0 \leq r \leq 1} | W_T(r) - W^*_T(r) | \leq \frac{|e_1| + |e_T|}{\sigma_e \sqrt{T}} = o_{\IP}(1).
    \]
    Since $e_1/\sigma_e$ and $e_T/\sigma_e$ are independent and both have distribution $F$, the claim follows from Slutsky's theorem.
\end{proof}
Note that this implies that 
\begin{equation}
\frac{X_{[Tr]}}{\sqrt{T}} \wto \sigma_e J_c(r), \quad 0 \leq r \leq 1,\label{eq:A.1}
\end{equation}
where $X_t$ is defined in the following lemma. Moreover, the convergence is joint with the boundary innovations. That is,
\[
\left(\frac{e_1}{\sigma_e}, \frac{e_T}{\sigma_e} , \frac{X_{[Tr]}}{\sigma_e \sqrt{T}}\right) \wto (\varepsilon_1, \varepsilon_{\infty}, J_c(r)).
\]
\begin{lemma}\label{lem:A1}
    \rev{Suppose \eqref{eq:model} and Assumption~\ref{ass:innovations} hold. Define $X_t=\rho_TX_{t-1}+e_t$, $X_0=0$, $a_T=(1-\delta/\sqrt T)/(1+c/T)$, and $b_T=(\delta/\sqrt T+c/T)/(1+c/T)$. Then $a_T\to1$, $\sqrt T b_T\to\delta$, and}
    \begin{align}
        &\text{(a)} \quad u_t = a_Te_t + b_TX_t, \label{eq:A.2} \\
        &\text{(b)} \quad \ty_t = \tilde{u}_t = a_T\tilde{e}_t + b_T\tilde{X}_t \label{eq:A.3}
    \end{align}
    where $\widetilde{\cdot}$ represents either GLS or OLS detrending being applied.
\end{lemma}

\begin{proof}
    (a) is a similar decomposition to \citet{NgPerron1996}, and (b) follows from linearity.
\end{proof}

\begin{lemma}\label{lem:A2}
    \rev{Suppose \eqref{eq:model} and Assumption~\ref{ass:innovations} hold. Let $\hat{\beta}_0(w)$ and $\hat{\beta}_1(w)$ be the OLS coefficients from regressing a series $w_t$ on $(1,t)$. Then}
    \begin{align}
        &\text{(a)} \quad \hat{\beta}_0(e) = O_{\IP}(T^{-1/2}), \quad \hat{\beta}_1(e) = O_{\IP}(T^{-3/2}), \label{eq:A.4} \\
        &\text{(b)} \quad \hat{\beta}_0(X) = O_{\IP}(T^{1/2}), \quad \hat{\beta}_1(X) = O_{\IP}(T^{-1/2}), \label{eq:A.5} \\
        &\text{(c)} \quad T^{-1/2}\tilde{X}_{\max\{1,[Tr]\}} \wto \sigma_e J_c^{\tau}(r), \quad 0 \leq r \leq 1. \label{eq:A.6}
    \end{align}
\end{lemma}

\begin{proof}
    For (a), write $\bar{t}=(T+1)/2$ and $\bar{e}=T^{-1}\sum_{t=1}^Te_t.$ Since $\sum_{t=1}^T(t-\bar{t})=0$,
    \[
    \hat{\beta}_1(e) = \frac{\sum_{t=1}^T(t-\bar{t})e_t}{\sum_{t=1}^T (t-\bar{t})^2}.
    \]
    The denominator equals $T(T^2-1)/12$ and is hence of order $T^3$. The numerator has variance
    \[
    \sigma^2_e\sum_{t=1}^T(t-\bar{t})^2 = O(T^3),
    \]
    so it is $O_{\IP}(T^{3/2})$. Hence $\hat{\beta}_1(e)=O_{\IP}(T^{-3/2})$. Since $\bar{e}=O_{\IP}(T^{-1/2})$,
    \[
    \hat{\beta}_0(e) = \bar{e}-\hat{\beta}_1(e)\bar{t} = O_{\IP}(T^{-1/2}),
    \]
    which proves \eqref{eq:A.4}. It follows immediately that
    \begin{equation}
    \sup_{1\leq t\leq T} |\hat{\beta}_0(e) + t\hat{\beta}_1(e)| \leq |\hat{\beta}_0(e)| + T|\hat{\beta}_1(e)| = O_{\IP}(T^{-1/2}) = o_{\IP}(1).\label{eq:A.7}
\end{equation}
    
    For (b), write $\bar X=T^{-1} \sum_{t=1}^TX_t$. \eqref{eq:A.1} implies $\max_{t\leq T}|X_t| = O_{\IP}(\sqrt{T}).$ Therefore, $\bar{X} = O_{\IP}(\sqrt{T}),$ and
    \[
    \left| \sum_{t=1}^T(t-\bar{t})(X_t-\bar{X})\right| \leq \max_t |t-\bar{t}|\sum_{t=1}^T|X_t-\bar{X}| = O_{\IP}(T^{5/2}).
    \]
    Division by the $O(T^3)$ denominator yields $\hat{\beta}_1(X) = O_{\IP}(T^{-1/2})$, and then $\hat{\beta}_0(X) = \bar{X}-\hat{\beta}_1(X)\bar{t} = O_{\IP}(T^{1/2}),$ proving \eqref{eq:A.5}.

    For (c), divide the two OLS normal equations for $X_t$ by $T^{3/2}$ and $T^{5/2},$ respectively. From \eqref{eq:A.1} and Riemann-sum convergence, we have
    \[
    T^{-3/2}\sum_{t=1}^TX_t \wto \sigma_e\int_0^1 J_c(r) dr, \quad T^{-5/2}\sum_{t=1}^T tX_t \wto \sigma_e\int_0^1 rJ_c(r) dr.
    \]
    Since $T^{-2}\sum_{t=1}^T t \to 1/2$ and $T^{-3}\sum_{t=1}^T t^2 \to 1/3,$ the rescaled normal equations converge to
    \[
    A+\frac{1}{2}B = \sigma_e \int_0^1 J_c(r) dr, \quad \frac{1}{2}A + \frac{1}{3}B = \sigma_e \int_0^1 rJ_c(r) dr,
    \]
    where $A$ and $B$ are the limits of $T^{-1/2}\hat{\beta}_0(X)$ and $T^{1/2}\hat{\beta}_1(X),$ respectively. Solving,
    \[
    B = 12\sigma_e \int_0^1 (r-\frac{1}{2})J_c(r)dr, \quad A=\sigma_e\int_0^1 J_c(r)dr - \frac{1}{2}B.
    \]
    These are exactly the coefficients of the continuous least squares projection used to define $J_c^{\tau}(r)$ in \eqref{eq:jcr_tau}. Hence,
    \[
    \frac{\tilde{X}_{\max\{1,[Tr]\}}}{\sqrt{T}} = \frac{X_{\max\{1,[Tr]\}}}{\sqrt{T}} - \frac{\hat{\beta}_0(X)}{\sqrt{T}} - (\sqrt{T}\hat{\beta}_1(X))\frac{\max\{1,[Tr]\}}{T} \wto \sigma_e J_c^{\tau}(r),
    \]
    which proves \eqref{eq:A.6}. This is the deterministic trend analogue of the sample-moment projection underlying Lemma~3.2 of \citet{NgPerron1996}.
\end{proof}

\begin{revision}
\begin{proof}[Proof of Proposition~\ref{prop:sm_ols}]
By Lemma~\ref{lem:A1}, OLS linearity gives
\[
 \ty_t=a_T\tilde e_t+b_T\tilde X_t.
\]
Lemma~\ref{lem:A2} implies
\[
 \frac1T\sum_{t=1}^{T-1}\tilde e_t^2=\sigma_e^2+o_{\IP}(1),\qquad
 \frac{\tilde X_{\max\{1,[Tr]\}}}{\sqrt T}\Rightarrow\sigma_eJ_c^\tau(r),
\]
To control the cross term, note that $\E X_{t-1}^2\le C T$ for $t\le T$. The martingale-difference property gives
\[
 \operatorname{Var}\!\left(\sum_{t=1}^T e_tX_{t-1}\right)
 =\sigma_e^2\sum_{t=1}^T\E X_{t-1}^2=O(T^2).
\]
Thus $\sum_{t=1}^T e_tX_t=\rho_T\sum_{t=1}^T e_tX_{t-1}+\sum_{t=1}^T e_t^2=O_{\IP}(T)$. Since $\sum_{t=1}^T e_t=O_{\IP}(T^{1/2})$ and $\sum_{t=1}^T te_t=O_{\IP}(T^{3/2})$, OLS orthogonality and Lemma~\ref{lem:A2} yield
\[
\begin{aligned}
 \sum_{t=1}^T\tilde e_t\tilde X_t
 &=\sum_{t=1}^T e_tX_t-\hat\beta_0(X)\sum_{t=1}^T e_t
   -\hat\beta_1(X)\sum_{t=1}^T te_t\\
 &=O_{\IP}(T).
\end{aligned}
\]
Removing the terminal product $\tilde e_T\tilde X_T=O_{\IP}(T^{1/2})$ therefore gives
\[
 \frac1T\sum_{t=1}^{T-1}\tilde e_t\tilde X_t=O_{\IP}(1).
\]
Therefore, using $a_T\to1$, $b_T=O(T^{-1/2})$, and $Tb_T^2\to\delta^2$,
\begin{align*}
 \frac1T\sum_{t=1}^{T-1}\ty_t^2
 &=a_T^2\frac1T\sum_{t=1}^{T-1}\tilde e_t^2
   +2a_Tb_T\frac1T\sum_{t=1}^{T-1}\tilde e_t\tilde X_t\\
 &\quad{}+Tb_T^2\frac1T\sum_{t=1}^{T-1}\left(\frac{\tilde X_t}{\sqrt T}\right)^2\\
 &\Rightarrow \sigma_e^2\left\{1+\delta^2\int_0^1[J_c^\tau(r)]^2\,dr\right\}.
\end{align*}
This is \eqref{eq:sm1_ols}. At the endpoint, Lemmas~\ref{A.0} and \ref{lem:A2} give $\tilde e_T=e_T+o_{\IP}(1)$ and $T^{-1/2}\tilde X_T\Rightarrow\sigma_eJ_c^\tau(1)$ jointly, hence
\[
 \ty_T=a_T\tilde e_T+b_T\tilde X_T
 \Rightarrow\sigma_e\{\varepsilon_\infty+\delta J_c^\tau(1)\},
\]
which is \eqref{eq:sm2_ols}.
\end{proof}
\end{revision}

\noindent The OLS argument above uses the fact that the line fitted to the noise component is uniformly negligible. Under GLS detrending, this is no longer true. Rather, the first quasi-differenced regressor has an order one first observation, while summation by parts in the trend score leaves a terminal boundary term. The next two lemmas isolate these effects before the sample moments are assembled.

\begin{lemma}\label{lem:A.3}
    \rev{Suppose \eqref{eq:model} and Assumption~\ref{ass:innovations} hold, with fixed $\cbar<0$. Let $Z^{\abar}=D_{\abar}Z$, where $\abar=1+\cbar/T$ and $D_{\abar}$ is the GLS quasi-difference operator matrix.} Write
    \[
    (Z^{\abar})'Z^{\abar} = \begin{pmatrix}
        m_{11} & m_{12} \\
        m_{12} & m_{22}
    \end{pmatrix}, \quad \Delta_T = m_{11}m_{22} - m^2_{12}.
    \]
    Then
    \begin{equation}
    m_{11} \to 1, \quad m_{12} \to 1-\cbar + \frac{\cbar^2}{2}, \quad \frac{m_{22}}{T} \to 1 - \cbar + \frac{\cbar^2}{3}, \quad \frac{\Delta_T}{T} \to 1 - \cbar + \frac{\cbar^2}{3}.\label{eq:A.15}
\end{equation}
    Furthermore, let $N_0(w)$ and $N_1(w)$ denote the two elements of $(Z^{\abar})'D_{\abar}w$. Then, jointly with the weak convergence described above,
    \begin{align}
        &N_0(e) = e_1 + o_{\IP}(1), \quad N_1(e) = (1-\cbar)e_T + o_{\IP}(1), \label{eq:A.16}\\
        &N_0(X) = e_1 + o_{\IP}(1), \quad \frac{N_1(X)}{\sqrt{T}} \wto \sigma_e\left[(1-\cbar)J_c(1)+\cbar^2\int_0^1rJ_c(r)dr \right]. \label{eq:A.17}
    \end{align}
\end{lemma}

\begin{proof}
    The two quasi-differenced deterministic regressors are 
    \[
    g_1 = 1, \quad g_t = -\frac{\cbar}{T} \quad t\geq 2,
    \]
    \[
    h_1=1, \quad h_t = 1 -\frac{\cbar(t-1)}{T} \quad t\geq 2.
    \]
    Consequently,
    \[
    m_{11} = 1 + \frac{\cbar^2(T-1)}{T^2} \to 1,
    \]
    \[
    m_{12} = 1 - \frac{\cbar(T-1)}{T} + \frac{\cbar^2(T-1)}{2T} \to 1 - \cbar + \frac{\cbar^2}{2},
    \]
    and,
    \[
    \frac{m_{22}}{T} = \frac{1}{T}\left[1+\sum_{j=1}^{T-1}\left(1-\frac{\cbar j}{T}\right)^2 \right] \to \int_0^1 (1-\cbar r)^2 dr = 1 - \cbar + \frac{\cbar^2}{3}.
    \]
    Since $m_{12} = O(1)$, the determinant limit in \eqref{eq:A.15} follows. 

    \revTwo{For the noise component $e$,}
    \[
    (D_{\abar}e)_1 = e_1, \quad (D_{\abar}e)_t = (e_t - e_{t-1}) - \frac{\cbar}{T}e_{t-1}, \quad t\geq 2.
    \]
    The first score is therefore
    \begin{equation}
    \begin{aligned}
    N_0(e) &= e_1 - \frac{\cbar}{T}\sum_{t=2}^T\left[(e_t-e_{t-1})-\frac{\cbar}{T}e_{t-1} \right] \\
    &= e_1 - \frac{\cbar}{T}(e_T-e_1)+\frac{\cbar^2}{T^2}\sum_{t=2}^Te_{t-1} = e_1 +o_{\IP}(1).
    \end{aligned}\label{eq:A.18}
\end{equation}
    As for $N_1(e),$ discrete summation by parts yields
    \[
    \sum_{t=2}^Th_t(e_t-e_{t-1}) = h_Te_T - h_2e_1 + \frac{\cbar}{T}\sum_{t=2}^{T-1}e_t.
    \]
    Hence,
    \begin{equation}
    N_1(e) = h_Te_T + (1-h_2)e_1 + \frac{\cbar}{T}\sum_{t=2}^{T-1}e_t - \frac{\cbar}{T}\sum_{t=2}^T h_t e_{t-1}.\label{eq:A.19}
\end{equation}
    The two sums in \eqref{eq:A.19} are $O_{\IP}(\sqrt{T})$, while $1-h_2 = \cbar/T$ and $h_T \to 1-\cbar$. Thus $N_1(e) = (1-\cbar)e_T + o_{\IP}(1)$. Equation \eqref{eq:A.19} explicitly shows how the boundary terms become involved. 

    \revTwo{For the near-integrated component $X$,}
    \[
    (D_{\abar}X)_1 = e_1, \quad (D_{\abar}X)_t = e_t + \frac{c-\cbar}{T}X_{t-1},\quad t \geq 2.
    \]
    Hence,
    \begin{equation}
    N_0(X) = e_1 - \frac{\cbar}{T}\sum_{t=2}^T e_t - \frac{\cbar(c-\cbar)}{T^2}\sum_{t=2}^T X_{t-1} = e_1 + o_{\IP}(1).\label{eq:A.20}
\end{equation}
    Furthermore, using $\sum_{t=1}^Te_t=O_{\IP}(\sqrt{T})$ and $\sum_{t=1}^TX_t = O_{\IP}(T^{3/2})$, we have
    \begin{equation}
    \frac{N_1(X)}{\sqrt{T}} = \frac{e_1}{\sqrt{T}} + \frac{1}{\sqrt{T}}\sum_{t=2}^Th_te_t + \frac{c-\cbar}{T^{3/2}}\sum_{t=2}^Th_tX_{t-1}.\label{eq:A.21}
\end{equation}
    The weighted noise limit stated above and \eqref{eq:A.1} yield jointly
    \[
    \frac{1}{\sqrt{T}}\sum_{t=2}^Th_te_t \wto \sigma_e\int_0^1 (1-\cbar r)dW(r),
    \]
    \[
    \frac{1}{T^{3/2}}\sum_{t=2}^Th_tX_{t-1} \wto \sigma_e \int_0^1 (1-\cbar r)J_c(r) dr.
    \]
    Using $dW(r) = dJ_c(r) - cJ_c(r)dr,$ the limit of \eqref{eq:A.21}, divided by $\sigma_e,$ is
    \[
    \int_0^1 (1-\cbar r) dJ_c(r) - \cbar \int_0^1 (1-\cbar r)J_c(r)dr.
    \]
    Integration by parts yields
    \[
    \int_0^1 (1-\cbar r) dJ_c(r) = (1-\cbar)J_c(1) + \cbar \int_0^1 J_c(r) dr,
    \]
    whereas
    \[
    -\cbar \int_0^1(1-\cbar r)J_c(r)dr = -\cbar \int_0^1J_c(r) dr + \cbar^2 \int_0^1 rJ_c(r) dr.
    \]
    The $\int J_c$ terms cancel, proving \eqref{eq:A.17}.
\end{proof}

\begin{lemma}\label{lem:A.4}
\rev{Suppose \eqref{eq:model} and Assumption~\ref{ass:innovations} hold, with fixed $\cbar<0$. Let $\hat{\psi}_0(w)$ and $\hat{\psi}_1(w)$ be the GLS coefficients obtained from \eqref{eq:gls_minimizer}, applied to a generic series $w_t$. Then, jointly with Lemma~\ref{A.0} and \eqref{eq:A.1},}
\begin{equation}
\hat{\psi}_0(e) \wto \sigma_e\varepsilon_1, \quad T\hat{\psi}_1(e) \wto \sigma_eq_{\cbar}\label{eq:A.22}
\end{equation}
and hence
\begin{equation}
\hat{\psi}_0(e) + [Tr]\hat{\psi}_1(e) \wto \sigma_e(\varepsilon_1 + q_{\cbar}r).\label{eq:A.23}
\end{equation}
\revTwo{For the near-integrated component $X$,}
\begin{equation}
\hat{\psi}_0(X) \wto \sigma_e\varepsilon_1\label{eq:A.24}
\end{equation}
\begin{equation}
\sqrt{T}\hat{\psi}_1(X) \wto \sigma_e\left[\lambda J_c(1) + 3 (1-\lambda)\int_0^1 sJ_c(s)ds \right],\label{eq:A.25}
\end{equation}
and therefore,
\begin{equation}
\frac{\tilde{X}_{\max\{1,[Tr]\}}}{\sqrt{T}} \wto \sigma_e V_{c,\cbar}(r).\label{eq:A.26}
\end{equation}
    
\end{lemma}

\begin{proof}
    The inverse of the GLS normal matrix yields
    \begin{equation}
    \hat{\psi}_0(w) = \frac{m_{22}N_0(w) - m_{12}N_1(w)}{\Delta_T}, \quad \hat{\psi}_1(w) = \frac{m_{11}N_1(w) - m_{12}N_0(w)}{\Delta_T}\label{eq:A.27}
\end{equation}
    From Lemma \ref{lem:A.3},
    \[
    \hat{\psi}_0(e) = \frac{(m_{22}/T)N_0(e)}{\Delta_T/T} - \frac{m_{12}N_1(e)}{\Delta_T} = e_1 + o_{\IP}(1),
    \]
    since $(m_{22}/T)/(\Delta_T/T) \to 1, N_0(e) = e_1 + o_{\IP}(1),$ and the second term is $O_{\IP}(T^{-1}).$ Similarly,
    \[
    T\hat{\psi}_1(e) = \frac{m_{11}N_1(e)-m_{12}N_0(e)}{\Delta_T/T} \wto \sigma_e\frac{(1-\cbar)\varepsilon_\infty - (1-\cbar +\cbar^2/2)\varepsilon_1}{1-\cbar+\cbar^2/3}.
    \]
    After division by $\sigma_e$ and the use of \eqref{eq:limit_constants},
    \[
    \frac{1-\cbar}{1-\cbar+\cbar^2/3} = \lambda, \quad \frac{1-\cbar+\cbar^2/2}{1-\cbar+\cbar^2/3} = \frac{3-\lambda}{2},
    \]
    so \eqref{eq:A.22} follows. \eqref{eq:A.23} then follows since $[Tr]/T \to r$ uniformly.

    \revTwo{For the near integrated component's, $X$'s, intercept}, \eqref{eq:A.27} and \eqref{eq:A.15}--\eqref{eq:A.17} give
    \[
    \hat{\psi}_0(X) = \frac{(m_{22}/T)N_0(X)}{\Delta_T/T} - \frac{m_{12}N_1(X)}{\Delta_T} = e_1 + o_{\IP}(1),
    \]
    since $(m_{22}/T)/(\Delta_T/T)\to 1, N_0(X) = e_1 + o_{\IP}(1), $ and the second term is $O_{\IP}(T^{-1/2}).$ For the slope,
    \[
    \begin{aligned}
    \sqrt{T}\hat{\psi}_1(X)
    &= \frac{m_{11}N_1(X)/\sqrt{T} - m_{12}N_0(X)/\sqrt{T}}{\Delta_T/T}\\
    &\wto \frac{\sigma_e}{1-\cbar+\cbar^2/3}\left[(1-\cbar)J_c(1) + \cbar^2 \int_0^1 rJ_c(r)dr \right].
    \end{aligned}
    \]
    Using $\lambda$ and $\cbar^2/(1-\cbar+\cbar^2/3)$ yields \eqref{eq:A.25}. Finally, 
    \[
    \frac{\tilde{X}_{\max\{1,[Tr]\}}}{\sqrt{T}} = \frac{X_{\max\{1,[Tr]\}}}{\sqrt{T}} - \frac{\hat{\psi}_0(X)}{\sqrt{T}} - (\sqrt{T}\hat{\psi}_1(X))\frac{\max\{1,[Tr]\}}{T}.
    \]
    The intercept term vanishes after division by $\sqrt{T}$, and substitution of \eqref{eq:A.1} and \eqref{eq:A.25} yields exactly $V_{c,\cbar}$ as defined in \eqref{gls_func}, proving \eqref{eq:A.26}. Joint convergence follows because every quantity above is obtained from the same jointly convergent boundary/path collection and Lemma \ref{lem:A.3} by continuous mapping and Slutsky's theorem.
\end{proof}

\begin{revision}
\begin{proof}[Proof of Proposition~\ref{prop:sm_gls}]
Let $l_t=\hat\psi_0(e)+t\hat\psi_1(e)$. Lemmas~\ref{lem:A1} and \ref{lem:A.4} yield the exact decomposition
\begin{equation}\label{eq:A_gls_decomposition}
 \ty_t=a_Te_t+\zeta_{T,t},\qquad
 \zeta_{T,t}=-a_Tl_t+b_T\tilde X_t.
\end{equation}
with, jointly in $D[0,1]$,
\begin{equation}\label{eq:A_gls_path}
 \zeta_{T,\max\{1,[Tr]\}}\Rightarrow
 \sigma_e\{\delta V_{c,\cbar}(r)-\varepsilon_1-q_{\cbar}r\}.
\end{equation}
Hence
\[
 \frac1T\sum_{t=1}^{T-1}\zeta_{T,t}^2
 \Rightarrow\sigma_e^2\int_0^1
 [\delta V_{c,\cbar}(r)-\varepsilon_1-q_{\cbar}r]^2\,dr.
\]
Moreover, Lemma~\ref{lem:A.4} gives $\hat\psi_0(e)=O_{\IP}(1)$, $\hat\psi_1(e)=O_{\IP}(T^{-1})$, $\hat\psi_0(X)=O_{\IP}(1)$ and $\hat\psi_1(X)=O_{\IP}(T^{-1/2})$. The innovation sums satisfy $T^{-1}\sum_{t=1}^{T-1}e_t=O_{\IP}(T^{-1/2})$ and $T^{-1}\sum_{t=1}^{T-1}te_t=O_{\IP}(T^{1/2})$. Together with the raw cross-moment bound proved for Proposition~\ref{prop:sm_ols}, these rates give
\[
 \frac1T\sum_{t=1}^{T-1}e_tl_t=O_{\IP}(T^{-1/2}),\qquad
 \frac1T\sum_{t=1}^{T-1}e_t\tilde X_t=O_{\IP}(1).
\]
Here the second bound follows by subtracting the GLS fitted line from $X_t$. Consequently, since $b_T=O(T^{-1/2})$,
\[
 \frac1T\sum_{t=1}^{T-1}e_t\zeta_{T,t}
 =-\frac{a_T}{T}\sum_{t=1}^{T-1}e_tl_t
   +\frac{b_T}{T}\sum_{t=1}^{T-1}e_t\tilde X_t
 =o_{\IP}(1).
\]
Since $T^{-1}\sum_{t=1}^{T-1} e_t^2\stackrel{\IP}{\to}\sigma_e^2$ and $a_T\to1$, expansion of $T^{-1}\sum_{t=1}^{T-1}\ty_t^2$ gives \eqref{eq:sm1_gls}.

At $t=T$,
\[
 \ty_T\Rightarrow\sigma_e\{\varepsilon_\infty+
 \delta V_{c,\cbar}(1)-\varepsilon_1-q_{\cbar}\}.
\]
Using $q_{\cbar}=\lambda\varepsilon_\infty-(3-\lambda)\varepsilon_1/2$ gives \eqref{eq:sm2_gls}.
\end{proof}
\end{revision}

\section{Appendix B}\label{sec:appendix_b}

\begin{revision}
Throughout this appendix, \eqref{eq:model} and Assumption~\ref{ass:innovations} hold together with the Gaussian oracle setup of Section~\ref{sec:env}. For a scalar $a$, let $D_a$ denote the lower-bidiagonal quasi-difference matrix satisfying
\[
 (D_aw)_1=w_1,\qquad (D_aw)_t=w_t-aw_{t-1},\quad t\ge2.
\]
Then $v=D_{-\theta_T}e$ and $\Sigma_T=\sigma_e^2D_{-\theta_T}D_{-\theta_T}'$.
\end{revision}

\begin{lemma}\label{lem:B.1}
    \rev{Under these conditions, define the exactly whitened stochastic component $x=D^{-1}_{-\theta_T}u$. Then} 
    \begin{equation}
    x_t = \rho_T x_{t-1}+e_t, \quad x_0 = 0.\label{eq:B.1}
\end{equation}
    Moreover, after profiling out the unknown intercept and trend, the Gaussian likelihood ratio statistic for $c=0$ against $\cbar$ can be written
    \begin{equation}
    L^{*,\tau}_T(\cbar) = A_T(\cbar) + Q_T(0) - Q_T(\cbar),\label{eq:B.2}
\end{equation}
    where
    \begin{equation}
    A_T(\cbar) = \frac{\cbar^2}{T^2\sigma^2_e}\sum_{t=1}^T x^2_{t-1} - \frac{2\cbar}{T\sigma^2_e}\sum_{t=1}^Tx_{t-1}\Delta x_t,\label{eq:B.3}
\end{equation}
    \rev{For a candidate $b\in\{0,\cbar\}$, let $\mathcal Z_T(b)=D_{1+b/T}D^{-1}_{-\theta_T}Z$ and define
    \[
      Q_T(b)=\sigma_e^{-2}(D_{1+b/T}x)'P_{\mathcal Z_T(b)}(D_{1+b/T}x),
      \qquad P_A=A(A'A)^{-1}A'.
    \]
    Thus $Q_T(b)$ is the regression sum of squares from projecting the whitened quasi-differenced stochastic component on the correspondingly transformed intercept and trend.} Finally,
    \begin{equation}
    A_T(\cbar) \wto \cbar^2 \int_0^1 J_c(r)^2dr - \cbar(J_c(1)^2-1).\label{eq:B.4}
\end{equation}
    
\end{lemma}

\begin{proof}
    The stochastic recursion is $D_{\rho_T}u = D_{-\theta_T}e$. Multiplication on the left and commutativity yield $D_{\rho_T}x=e$, proving \eqref{eq:B.1}.

    For a candidate $\abar$, the Gaussian criterion is proportional to
    \[
    \|D^{-1}_{-\theta_T}D_{\abar}(y-Z\gamma)\|^2.
    \]
    Thus exact whitening converts the stochastic component into $x$, while the regressors become $D_{\abar}D^{-1}_{-\theta_T}Z.$ Profiling over $\gamma$ by OLS gives \eqref{eq:B.2}.

    From \eqref{eq:B.1},
    \[
    (D_{\abar}x)_t = x_t - (1+\frac{\cbar}{T})x_{t-1} = \Delta x_t - \frac{\cbar}{T}x_{t-1}.
    \]
    Subtracting the corresponding criterion under $\cbar = 0$ yields \eqref{eq:B.3}. 

    \eqref{eq:A.1} implies
    \begin{equation}
    \frac{1}{T^2\sigma^2_e}\sum_{t=1}^T x^2_{t-1} \wto \int_0^1 J_c(r)^2dr.\label{eq:B.5}
\end{equation}
    Also, $2x_{t-1} \Delta x_{t} = x^2_t - x^2_{t-1} - (\Delta x_t)^2$. Hence
    \begin{equation}
    \frac{2}{T\sigma^2_e}\sum_{t=1}^T x_{t-1}\Delta x_t = \frac{x^2_T}{T\sigma^2_e} - \frac{1}{T\sigma^2_e}\sum_{t=1}^T (\Delta x_t)^2.\label{eq:B.6}
\end{equation}
    The first term converges to $J_c(1)^2$. Since $\Delta x_t = e_t + (c/T)x_{t-1}$, the law of large numbers and the usual local-to-unity moment bounds yield 
    \[
    \frac{1}{T\sigma^2_e}\sum_{t=1}^T (\Delta x_t)^2 \stackrel{\IP}{\to} 1.
    \]
    \eqref{eq:B.3}, \eqref{eq:B.5}, and \eqref{eq:B.6} establish \eqref{eq:B.4}.
\end{proof}

\begin{lemma}\label{lem:B.2}
    \rev{Under \eqref{eq:model}, Assumption~\ref{ass:innovations}, the Gaussian setup of Section~\ref{sec:env}, and fixed $\cbar<0$, let $g_t$ and $h_t$ denote the intercept and trend columns of $D_{\abar}D^{-1}_{-\theta_T}Z$, respectively. Then}
    \begin{equation}
    g_t = (-\theta_T)^{t-1} - \frac{\cbar}{T}\frac{1-(-\theta_T)^{t-1}}{1+\theta_T},\label{eq:B.7}
\end{equation}
    \begin{equation}
    h_t = \sum_{j=1}^t (-\theta_T)^{t-j}\left[1-\frac{\cbar(j-1)}{T} \right].\label{eq:B.8}
\end{equation}
    Their quadratic products satisfy
    \begin{equation}
    T^{-1/2}\sum_{t=1}^T g^2_t \to \frac{1}{2\delta},\label{eq:B.9}
\end{equation}
    \begin{equation}
    T^{-2}\sum_{t=1}^T h^2_t \to \frac{1}{\delta^2}\left(1-\cbar + \frac{\cbar^2}{3} \right),\label{eq:B.10}
\end{equation}
    and
    \begin{equation}
    T^{-5/4}\sum_{t=1}^T g_th_t \to 0.\label{eq:B.11}
\end{equation}
    Moreover, there exists a random sequence $\zeta_T = O_{\IP}(1)$, common to the null and the point alternative, such that
    \begin{equation}
    \frac{1}{T^{1/4}\sigma_e}\sum_{t=1}^T g_t (D_{\abar}x)_t = \zeta_T + o_{\IP}(1),\label{eq:B.12}
\end{equation}
    whereas
    \begin{equation}
    \frac{1}{T\sigma_e}\sum_{t=1}^T h_t (D_{\abar}x)_t \wto \frac{1}{\delta}\left[(1-\cbar)J_c(1) +\cbar^2 \int_0^1 rJ_c(r)dr \right].\label{eq:B.13}
\end{equation}
\rev{The conclusions hold jointly for the two candidates $b\in\{0,\cbar\}$, with $\cbar$ replaced by $b$ and $D_{\abar}$ by $D_{1+b/T}$ throughout; the intercept-score approximation uses the same $\zeta_T$ for both candidates.}
\end{lemma}

\begin{proof}
    Quasi-differencing the intercept gives 1 in the first observation and $-\cbar/T$ after. Since $(D^{-1}_{-\theta_T}q)_t = \sum_{j=1}^t (-\theta_T)^{t-j}q_j,$ \eqref{eq:B.7} follows immediately. Quasi-differencing the trend gives $(D_{\abar}t)_j = 1 - \frac{\cbar(j-1)}{T}$, yielding \eqref{eq:B.8}. Since $-\theta_T = 1 - \delta/\sqrt{T}$ and $1+\theta_T = \delta/\sqrt{T},$ \eqref{eq:B.7} can be written
    \begin{equation}
    g_t = (-\theta_T)^{t-1} - \frac{\cbar}{\delta\sqrt{T}}[1-(-\theta_T)^{t-1}].\label{eq:B.14}
\end{equation}
    The second term is uniformly $O(T^{-1/2})$, while $(-\theta_T)^t = (1-\delta/\sqrt{T})^t$ is appreciable only over the first $O(\sqrt{T})$ observations. Hence
    \[
    T^{-1/2} \sum_{t=1}^T g^2_t = T^{-1/2}\sum_{t=1}^T (-\theta_T)^{2(t-1)} + o(1) \to \frac{1}{2\delta},
    \]
    proving \eqref{eq:B.9}.

    For the trend column, the geometric-sum representation in \eqref{eq:B.8} gives, for every fixed $r>0$,
    \begin{equation}
    \frac{h_{[Tr]}}{\sqrt{T}} \to \frac{1-\cbar r}{\delta}.\label{eq:B.15}
\end{equation}
    Moreover, $|h_t| \leq C \min(t, \sqrt{T})$ uniformly in $t$. We prove \eqref{eq:B.10} since
    \[
    T^{-2} \sum_{t=1}^T  h^2_t \to \frac{1}{\delta^2} \int_0^1 (1-\cbar r)^2 dr,
    \]
    by Riemann summation. The bounds in \eqref{eq:B.14}--\eqref{eq:B.15} similarly imply $\sum_{t=1}^T g_t h_t = O(T)$, and hence \eqref{eq:B.11}.

    For the scores, note from \eqref{eq:B.1} that 
    \begin{equation}
    (D_{\abar}x)_t = e_t + \frac{c-\cbar}{T}x_{t-1}\label{eq:B.16}
\end{equation}
    Define
    \begin{equation}
    \zeta_T = \frac{1}{T^{1/4}\sigma_e}\sum_{t=1}^T (-\theta_T)^{t-1}e_t.\label{eq:B.17}
\end{equation}
    Its variance converges to $(2\delta)^{-1}$, so $\zeta_T=O_{\IP}(1)$.
    By \eqref{eq:B.14}, replacing $g_t$ by $(-\theta_T)^{t-1}$
    in the normalized innovation score has an error with variance
    $O(T^{-1/2})$. Also, $\sum_{t=1}^T|g_t|=O(\sqrt T)$ and
    $\max_{t\leq T}|x_t|=O_{\IP}(\sqrt T)$, so the path contribution in
    \eqref{eq:B.16}, divided by $T^{1/4}\sigma_e$, is $O_{\IP}(T^{-1/4})$.
    These bounds hold for both candidates $0$ and $\cbar$.
    Thus \eqref{eq:B.12} holds simultaneously with the same $\zeta_T$,
    which does not depend on the candidate.

    For the second score, \eqref{eq:B.15}, \eqref{eq:A.1}, and \eqref{eq:B.16} yield
    \begin{equation}
    \frac{1}{T\sigma_e}\sum_{t=1}^T h_t (D_{\abar}x)_t \wto \frac{1}{\delta}\left[\int_0^1 (1-\cbar r)dW(r) + (c-\cbar)\int_0^1 (1-\cbar r)J_c(r) dr \right].\label{eq:B.18}
\end{equation}
    Using $dW(r) = dJ_c(r) - cJ_c(r) dr, $ the term in brackets becomes $\int_0^1(1-\cbar r)dJ_c(r) - \cbar \int_0^1 (1-\cbar r)J_c(r)dr.$ Integration by parts yields $\int_0^1 (1-\cbar r)dJ_c(r) = (1-\cbar)J_c(1) + \cbar \int_0^1 J_c(r) dr$, and the two terms involving $\int J_c$ cancel. The remaining expression is
    \[
    (1-\cbar) J_c(1) + \cbar^2 \int_0^1 rJ_c(r) dr,
    \]
    which proves \eqref{eq:B.13}.
\end{proof}

\begin{proof}[Proof of Theorem \ref{thm:power_envelope}]
We normalize the intercept column by $T^{-1/4}$ and the trend column by $T^{-1}$ of $D_{\abar}D_{-\theta_T}^{-1}Z.$ By \eqref{eq:B.9}--\eqref{eq:B.11}, the corresponding Gram matrix satisfies
\[
\begin{pmatrix}
    T^{-1/2}\sum_{t=1}^Tg^2_t & T^{-5/4}\sum_{t=1}^T g_t h_t \\
    T^{-5/4}\sum_{t=1}^T g_t h_t & T^{-2}\sum_{t=1}^T h^2_t
\end{pmatrix} \to \begin{pmatrix}
    (2\delta)^{-1} & 0 \\
    0 & \delta^{-2}(1-\cbar +\cbar^2/3)
\end{pmatrix}
\]
Apply these Gram-matrix limits simultaneously to the candidates
$b\in\{0,\cbar\}$, writing $g_t(b),h_t(b)$ for the columns in
\eqref{eq:B.7}--\eqref{eq:B.8} with $\cbar$ replaced by $b$.
Their normalized intercept scores equal the same
$\zeta_T+o_{\IP}(1)$, by \eqref{eq:B.12}.
For the trend score, geometric summation in \eqref{eq:B.8} gives,
for all sufficiently large $T$,
\[
 \left|\frac{h_t(b)}{\sqrt T}
       -\frac{1-b(t-1)/T}{\delta}\right|
 \leq C\{(-\theta_T)^t+T^{-1/2}\},
 \qquad 1\leq t\leq T.
\]
The resulting innovation-score error has variance $O(T^{-1/2})$.
Its path part is $O_{\IP}(T^{-1/2})$, using
$\max_{t\leq T}|x_t|=O_{\IP}(\sqrt T)$ and
$\sum_{t=1}^T(-\theta_T)^t=O(\sqrt T)$.
Discrete summation by parts, with $x_0=0$, therefore yields
\[
 \frac{1}{T\sigma_e}\sum_{t=1}^T h_t(b)(D_{1+b/T}x)_t
 =\frac{1}{\delta\sigma_e\sqrt T}
 \left[\left\{1-\frac{b(T-1)}T\right\}x_T
       +\frac{b^2}{T^2}\sum_{t=1}^{T-1}t x_t\right]+o_{\IP}(1),
\]
simultaneously for the two candidates.
Both normalized scores are bounded in probability. Inverting the
Gram matrices consequently gives the finite-sample approximations
\begin{equation}
 Q_T(\cbar)=2\delta\zeta_T^2+
 \frac{\left[
 \left(1-\cbar+\frac{\cbar}{T}\right)\frac{x_T}{\sigma_e\sqrt T}
 +\frac{\cbar^2}{\sigma_e T^{5/2}}\sum_{t=1}^{T-1}t x_t
 \right]^2}{1-\cbar+\cbar^2/3}
 +o_{\IP}(1).\label{eq:B.19}
\end{equation}
For the null candidate,
\begin{equation}
 Q_T(0)=2\delta\zeta_T^2+\frac{x_T^2}{T\sigma_e^2}
 +o_{\IP}(1).\label{eq:B.20}
\end{equation}
After subtracting $2\delta\zeta_T^2$, both approximations depend on
the same normalized path. The functional limit theorem and the
moment bounds used in \eqref{eq:B.5}--\eqref{eq:B.6} give their
joint weak convergence with $A_T(\cbar)$.
The common boundary contribution cancels before taking limits:
\begin{equation}
Q_T(0) - Q_T(\cbar) \wto J_c(1)^2 - \frac{\left[(1-\cbar)J_c(1) + \cbar^2 \int_0^1 rJ_c(r)dr \right]^2}{1-\cbar + \cbar^2/3}.\label{eq:B.21}
\end{equation}

Combining Lemma~\ref{lem:B.1} with \eqref{eq:B.21} yields
\begin{equation}
L^{*,\tau}_T(\cbar) \wto \cbar + \cbar^2 \int_0^1 J_c(r)^2 dr + (1-\cbar)J_c(1)^2 -\frac{\left[(1-\cbar)J_c(1) + \cbar^2 \int_0^1 rJ_c(r)dr \right]^2}{1-\cbar + \cbar^2/3}.\label{eq:B.22}
\end{equation}
Recall $\lambda = (1-\cbar)/(1-\cbar+\cbar^2/3)$ and $V_{c,\cbar}(r) = J_c(r) - r[\lambda J_c(1) + 3(1-\lambda) \int_0^1 sJ_c(s)ds].$ Since $3(1-\lambda)=\cbar^2/(1-\cbar+\cbar^2/3)$, the coefficient multiplying $r$ in $V_{c,\cbar}(r)$ is 
\[
\frac{(1-\cbar)J_c(1) + \cbar^2 \int_0^1 sJ_c(s) ds}{1-\cbar + \cbar^2/3}
\]
Expanding $V_{c,\cbar}(r)$ gives
\begin{equation}
\begin{aligned}
\cbar^2 \int_0^1 V_{c,\cbar}(r)^2dr + (1-\cbar)V_{c,\cbar}(1)^2&= \cbar^2 \int_0^1 J_c(r)^2dr + (1-\cbar)J_c(1)^2\\
&\quad{}-\frac{\left[(1-\cbar)J_c(1) + \cbar^2 \int_0^1 rJ_c(r)dr \right]^2}{1-\cbar + \cbar^2/3}.
\end{aligned}\label{eq:B.23}
\end{equation}
Thus $L^{*,\tau}_T(\cbar) \wto \cbar + P^{\tau}_{c,\cbar}$ where $P^{\tau}_{c,\cbar} = \cbar^2 \int_0^1 V_{c,\cbar}(r)^2dr + (1-\cbar)V_{c,\cbar}(1)^2.$

Notice that $\delta$ has disappeared. In particular, the limiting distribution of $L_T^{*,\tau}(\bar c)$ does not depend on $\delta$, since neither $V_{c,\bar c}$ nor the coefficients defining $P^\tau_{c,\bar c}$ involve $\delta$.  Let $k^{\tau}_{\alpha}(\cbar)$ denote the lower $\alpha$-quantile of $P^{\tau}_{0,\cbar}$. The asymptotic power of the point optimal invariant test indexed by $\cbar$ against a true local parameter $c$ is therefore
\begin{equation}
\pi^{\tau}(c,\cbar;\alpha)=\Pr(P^{\tau}_{c,\cbar} < k_{\alpha}^{\tau}(\cbar)).\label{eq:B.24}
\end{equation}
The pointwise envelope of \citet{ERS1996} is obtained by setting $\cbar = c$. Hence
\begin{equation}
\Pi^{\tau}_{NINW}(c,\delta;\alpha) = \Pr\left(c^2\int_0^1 V_{c,c}(r)^2dr + (1-c)V_{c,c}(1)^2 < k^{\tau}_{\alpha}(c) \right),\label{eq:B.25}
\end{equation}
which is exactly the standard envelope of \citet{ERS1996}. This completes the proof.

\end{proof}

\section{Appendix C}\label{sec:appendix_c}
\begingroup
\setlength{\jot}{3pt}
\setlength{\abovedisplayskip}{5pt plus 2pt minus 2pt}
\setlength{\belowdisplayskip}{5pt plus 2pt minus 2pt}
\setlength{\abovedisplayshortskip}{3pt plus 2pt}
\setlength{\belowdisplayshortskip}{4pt plus 2pt minus 2pt}
\setlength{\emergencystretch}{1em}

We prove Proposition~\ref{prop:feasible_s2ar} by applying an invertible transformation to
the lagged differences. Lemma~\ref{C:glsinputs} gives the required GLS bounds, and
Lemmas~\ref{C:innovation}--\ref{C:samplelemma} establish the coefficient
and cross-product approximations for the transformed regression.
The normal equations then give the residual-variance limit and,
through \eqref{C:coeflimit}, the limit
$\delta^{-1}\tanh(\kappa\delta/2)$ for $(1-\hat b(1))/\sqrt T$.

Throughout, the assumptions and notation of Proposition~\ref{prop:feasible_s2ar} apply.
In particular, $k$ is deterministic and satisfies
$k/\sqrt T\to\kappa\in(0,\infty)$.
Set $n=T-k-1$. Regression vectors and regressor matrices are stacked
over $t=k+2,\ldots,T$, with date $t$ corresponding to row $t-k-1$.
Dependence of the sample matrices and vectors on $T$ and $k$ is suppressed.
Write $\hat\psi_j=\hat\psi_j(u)$, $j=0,1$, for the GLS coefficients
from \eqref{eq:gls_minimizer} applied to $u_t$.
By \eqref{eq:lrv} and \eqref{eq:s^2_ar}, it suffices to establish
\begin{equation}\label{C:targets}
 \hat\sigma_{ek}^2\xrightarrow{\IP}\sigma_e^2,\qquad
 \frac{1-\hat b(1)}{\sqrt T}
 \xrightarrow{\IP}\frac1\delta\tanh(\kappa\delta/2).
\end{equation}
Normalize $\sigma_e^2=1$ until the final paragraph.
Vector norms are Euclidean and matrix norms are induced operator norms.
For a $d_1\times d_2$ matrix $D$,
$\|D\|_{\mathrm F}^2=\sum_{p=1}^{d_1}\sum_{q=1}^{d_2}D_{pq}^2$. The constant $C$ is finite,
independent of $T,k$, and may change between bounds. For any dimension
$m$, $\mathbf1_m$ is the vector of ones and $I_m$ is the identity matrix.
Inverses are used on events whose probabilities tend to one; the
associated quantities may be defined arbitrarily elsewhere.

Recall that model~\eqref{eq:model} is
\[
 u_t=\rho_Tu_{t-1}+e_t+\theta_Te_{t-1},\qquad
 \rho_T=1+c/T,\qquad \theta_T=-1+\delta/\sqrt T.
\]
Lemma~\ref{lem:A1} also gives $u_t=e_t+(\rho_T+\theta_T)X_{t-1}$, where
$X_t=\rho_TX_{t-1}+e_t$ and $X_0=0$. Define, for $1\leq t\leq T$,
\begin{equation}\label{C:leveldecomp}
 \xi_t=(\rho_T+\theta_T)X_{t-1}-\hat\psi_0-t\hat\psi_1,
 \qquad \ty_t=e_t+\xi_t.
\end{equation}
The corresponding $n$-vector for the lagged level is
$\xi=(\xi_{k+1},\ldots,\xi_{T-1})'$.

\begin{lemma}\label{C:glsinputs}
Suppose model~\eqref{eq:model} and Assumption~\ref{ass:innovations} hold, with fixed $\cbar<0$
and $\sigma_e^2=1$. Let $k=k_T$ be a deterministic integer sequence
such that $k/\sqrt T\to\kappa\in(0,\infty)$. Then the GLS
coefficients satisfy $\hat\psi_0=\Op(1)$ and
$\hat\psi_1=\Op(T^{-1})$. Moreover, $\max_{k+1\leq t\leq T-1}|\xi_t|=\Op(1)$ and
\begin{equation}\label{C:energy}
 x_T=\frac{\xi'\xi}{n}\wto
 x=\int_0^1\{\delta V_{c,\cbar}(r)-\varepsilon_1-q_{\cbar}r\}^2\,dr,
 \qquad \Pr(x>0)=1,
\end{equation}
where the limiting objects are defined in Section~\ref{sec:oracle}.
\end{lemma}

\begin{proof}
Lemma~\ref{lem:A1} gives $u_t=a_Te_t+b_TX_t$, with $a_T=O(1)$,
$b_T=O(T^{-1/2})$ and $a_T+b_T=1$. Linearity and Lemma~\ref{lem:A.4} imply
\[
 \hat\psi_0=a_T\hat\psi_0(e)+b_T\hat\psi_0(X)=\Op(1),\qquad
 \hat\psi_1=a_T\hat\psi_1(e)+b_T\hat\psi_1(X)=\Op(T^{-1}).
\]
Recall from \eqref{eq:A_gls_decomposition}, in the proof of
Proposition~\ref{prop:sm_gls}, that $\ty_t=a_Te_t+\zeta_{T,t}$, where
$\zeta_{T,t}=b_TX_t-\hat\psi_0-t\hat\psi_1$, and
\eqref{eq:A_gls_path} gives weak convergence of
$\zeta_{T,\max\{1, [Tr]\}}$ to
$\delta V_{c,\cbar}(r)-\varepsilon_1-q_{\cbar}r$.
Equation~\eqref{C:leveldecomp} gives $\xi_t=\zeta_{T,t}-b_Te_t$.
The fourth-moment assumption yields
\[
 \max_{k+1\leq t\leq T-1}|\xi_t-\zeta_{T,t}|
 \leq |b_T|\max_{1\leq t\leq T}|e_t|
 =\Op(T^{-1/2}T^{1/4})=\op(1).
\]
Continuous mapping and Riemann summation establish the maximum bound
and \eqref{C:energy}; deleting $k+1=o(T)$ dates changes the normalized
sum of squares by $\op(1)$. The limiting integrand before squaring is
$\delta J_c$ minus a random affine function. Its quadratic variation
along dyadic partitions is $\delta^2>0$ almost surely, so it cannot
vanish identically. Hence $\Pr(x>0)=1$.
\end{proof}

We next express the autoregression in Proposition~\ref{prop:feasible_s2ar} in terms of
transformed lagged differences. The fitted regression in \eqref{eq:ADF} is
\[
 \Delta\ty_t=\hat b_0\ty_{t-1}
       +\sum_{j=1}^k\hat b_j\Delta\ty_{t-j}+\hat e_{t,k},
 \qquad k+2\leq t\leq T,
\]
where $\hat e_{t,k}$ is the fitted residual and
$\hat b(1)=\sum_{j=1}^k\hat b_j$.
Following the transformation of
\citet[p.~1545 and eqns.~(A.1)--(A.3), p.~1546]{NgPerron2001},
define, for $k+2\leq t\leq T$ and $1\leq j\leq k$,
\begin{equation}\label{C:transform}
 \begin{aligned}
 z_{t,j}&=\sum_{m=0}^{k-j}(-\theta_T)^m\Delta\ty_{t-j-m},\\
 R_j&=\sum_{m=0}^{k-j}(-\theta_T)^m
     =\frac{1-(-\theta_T)^{k-j+1}}{1+\theta_T}.
 \end{aligned}
\end{equation}
The same formula defines $R_0$ at $j=0$. Let $Z\in\mathbb R^{n\times k}$
have row $(z_{t,1},\ldots,z_{t,k})$ at date $t$, and define
$\mu=(\mu_1,\ldots,\mu_k)'$, $\mu_j=(-\theta_T)^{k-j+1}$,
and $\gamma=\mathbf1_k-\mu$.
The geometric weights in \eqref{C:transform} cancel the intermediate
innovations, leaving $e_{t-j}-\mu_je_{t-k-1}$ together with the local
autoregressive and detrending terms displayed in \eqref{C:lagidentity}.
The common terms $e_{t-k-1}$ yield a rank-one decomposition of the auxiliary
cross-product matrix, with the remaining matrix converging to $I_k$
in operator norm; see \eqref{C:rankone}.

The transformation is triangular with unit diagonal, so it preserves
the regressor space and residual sum of squares. If
$\hat\beta=(\hat\beta_1,\ldots,\hat\beta_k)'\in\mathbb R^k$ denotes the coefficients
on $Z$ in the regression also containing the lagged level, then
\begin{equation}\label{C:coeffmap}
 \begin{aligned}
 \hat b_m&=\sum_{j=1}^m(-\theta_T)^{m-j}\hat\beta_j,
       \qquad 1\leq m\leq k,\\
 (1+\theta_T)\hat b(1)&=\gamma'\hat\beta.
 \end{aligned}
\end{equation}
Since $\ty_t=u_t-\hat\psi_0-t\hat\psi_1$, model~\eqref{eq:model} gives, for $t\geq2$,
\[
 \Delta\ty_t=\frac cT u_{t-1}+e_t+\theta_Te_{t-1}-\hat\psi_1.
\]
Substitution into \eqref{C:transform} gives
\begin{equation}\label{C:lagidentity}
 z_{t,j}=e_{t-j}-\mu_je_{t-k-1}+p^c_{t,j}-\hat\psi_1R_j,
\end{equation}
where
\begin{equation}\label{C:perturbations}
 p^c_{t,j}=\frac cT\sum_{m=0}^{k-j}(-\theta_T)^mu_{t-j-m-1}.
\end{equation}
Using the case $j=1$ and $R_0=1-\theta_TR_1$, the response satisfies
\begin{equation}\label{C:response}
 \Delta\ty_t=(c/T)\ty_{t-1}+\theta_Tz_{t,1}+r_t,
\end{equation}
with
\begin{equation}\label{C:responseparts}
 \begin{aligned}
 r_t&=e_t-(-\theta_T)^{k+1}e_{t-k-1}-\theta_Tp^c_{t,1}+r^d_t,\\
 r^d_t&=\frac cT\{\hat\psi_0+(t-1)\hat\psi_1\}-\hat\psi_1R_0.
 \end{aligned}
\end{equation}
The superscripts $c,d$ distinguish the local autoregressive and
detrending terms.
All level observations in these identities have dates at least one.

We first study the regression with regressors
$e_{t-j}-\mu_je_{t-k-1}$ and response
$e_t-(-\theta_T)^{k+1}e_{t-k-1}$. Let $G$ be the $n\times k$ matrix
\[
 G=\begin{pmatrix}
 e_{k+1}&e_k&\cdots&e_2\\
 e_{k+2}&e_{k+1}&\cdots&e_3\\
 \vdots&\vdots&\ddots&\vdots\\
 e_{T-1}&e_{T-2}&\cdots&e_{T-k}
 \end{pmatrix},
\]
whose row at date $t$ is $(e_{t-1},\ldots,e_{t-k})$. The $n$-vectors
$\eta=(e_1,\ldots,e_n)'$ and $f=(e_{k+2},\ldots,e_T)'$
have corresponding entries $e_{t-k-1}$ and $e_t$.
Define the $n\times k$ auxiliary regressor matrix $Z_0$ and the
$n$-vector $r_0$ by
\begin{equation}\label{C:innovationarrays}
 Z_0=G-\eta\mu',\qquad r_0=f-(-\theta_T)^{k+1}\eta.
\end{equation}
For the regression of $r_0$ on $Z_0$, define the cross-product
matrix $Q_0\in\mathbb R^{k\times k}$ and coefficient vector
$\alpha_0\in\mathbb R^k$ by
\[
 Q_0=\frac{Z_0'Z_0}{n},\qquad
 \alpha_0=Q_0^{-1}\frac{Z_0'r_0}{n}.
\]
Write $\mathbf e_1=(1,0,\ldots,0)'\in\mathbb R^k$. The next lemma
gives the coefficient and fitted-value bounds for
this regression. The common terms involving $e_{t-k-1}$ in the response
and regressors determine the leading term of $\gamma'\alpha_0$.

\begin{lemma}\label{C:innovation}
Suppose Assumption~\ref{ass:innovations} holds with $\sigma_e^2=1$, and let
$\theta_T=-1+\delta/\sqrt T$. Let $k=k_T$ be a deterministic integer
sequence such that $k/\sqrt T\to\kappa\in(0,\infty)$.
Then the matrix $Q_0$ is positive definite with probability tending to
one, and $\|Q_0^{-1}\|=\Op(1)$. Also,
\begin{equation}\label{C:auxbounds}
 \|\alpha_0\|=\Op(T^{-1/4}),\qquad
 \mathbf e_1'\alpha_0=\Op(T^{-3/8}),\qquad
 \mathbf e_1'Q_0^{-1}\mathbf e_1=1+\op(1),
\end{equation}
referred to respectively as the \emph{norm bound}, \emph{leading-entry
bound}, and \emph{diagonal limit}; and
\begin{equation}\label{C:innovationlimits}
 \gamma'\alpha_0=(-\theta_T)^{k+1}
       \frac{\gamma'\mu}{1+\|\mu\|^2}+\op(1),\qquad
 \frac{r_0'Z_0\alpha_0}{n}=(-\theta_T)^{2k+2}
       \frac{\|\mu\|^2}{1+\|\mu\|^2}+\op(1).
\end{equation}
referred to as the \emph{coefficient-sum limit} and the
\emph{fitted-value limit}.
\end{lemma}

\begin{proof}
We first bound the sample cross products of the building-block
matrices $G$, $\eta$, and $f$; the bounds
\eqref{C:gram}--\eqref{C:directionvariance} established here are used
throughout the argument.
For a deterministic $m\times m$ matrix $D$ and a vector
$\epsilon=(\epsilon_1,\ldots,\epsilon_m)'$ of iid unit-variance, mean-zero
innovations, independence and the fourth-moment assumption imply
\begin{equation}\label{C:quadratic}
 \E\{(\epsilon'D\epsilon-\operatorname{tr} D)^2\}\leq C\|D\|_{\mathrm F}^2.
\end{equation}
For nonsymmetric $D$, apply this inequality to $(D+D')/2$.
For the following comparison only, extend the innovations to an iid
sequence $e_t^*$ on the integers, with $e_t^*=e_t$ for $t\geq1$.
The auxiliary variables at $t\leq0$ do not alter the model's
initialization $e_0=0$. For
$\hat c_j=n^{-1}\sum_{t=1}^n e_t^*e_{t-j}^*$, $0\leq j<k$,
define the $k\times k$ Toeplitz matrix
$B=[\hat c_{|j-l|}]_{j,l=1}^k$.
Each entry of $G'G/n-B$ contains at most $2k$ endpoint products,
with cancelling diagonal means. Equation~\eqref{C:quadratic} gives
$\E\|G'G/n-B\|_{\mathrm F}^2\leq Ck^3/n^2$.
To bound $B-\hat c_0I_k$, define, for $\lambda\in[0,2\pi]$,
\[
 p(\lambda)=2\sum_{j=1}^{k-1}\hat c_j\cos(j\lambda)
 =\frac2n\Re\sum_{t=1}^n e_t^*L_t(\lambda),\qquad
 L_t(\lambda)=\sum_{j=1}^{k-1}e^{\mathrm i j\lambda}e_{t-j}^*.
\]
For fixed $\lambda$, $e_t^*L_t(\lambda)$ is a martingale difference,
$e_t^*$ is independent of $L_t(\lambda)$, and
$\E|L_t(\lambda)|^4\leq Ck^2$. The fourth-moment inequality of \citet{Burkholder1973}, applied to the real and imaginary parts, yields
\[
 \E|p(\lambda)|^4
 \leq\frac C{n^4}\E\left(\sum_{t=1}^n
 |e_t^*L_t(\lambda)|^2\right)^2\leq Ck^2/n^2.
\]
Bernstein's derivative inequality, $\sup|p'(\cdot)|\leq(k-1)\sup|p(\cdot)|$,
bounds $\sup|p|$ by twice its maximum on a grid of $O(k)$ points.
The union and Markov inequalities imply
$\Pr\{\sup|p|>\varepsilon\}\leq Ck^3/(n^2\varepsilon^4)$.
The Toeplitz inequality
$\|B-\hat c_0I_k\|\leq\sup_{\lambda\in[0,2\pi]}|p(\lambda)|$,
together with $\hat c_0-1=\Op(n^{-1/2})$, therefore gives
\begin{equation}\label{C:gram}
 \|G'G/n-I_k\|
 =\Op\!\left(n^{-1/2}+k^{3/2}/n+(k^3/n^2)^{1/4}\right)
 =\Op(T^{-1/8}).
\end{equation}
Independence and \eqref{C:quadratic} also give
\begin{equation}\label{C:cross}
\begin{aligned}
 \|G'\eta/n\|+\|G'f/n\|&=\Op(T^{-1/4}),&
 \eta'\eta/n&=1+\Op(T^{-1/2}),\\
 |\mathbf e_1'G'\eta/n|+|\mathbf e_1'G'f/n|&=\Op(T^{-1/2}),&
 \eta'f/n&=\Op(T^{-1/2}).
\end{aligned}
\end{equation}
For $\nu\in\{\mu,\gamma\}$, application of \eqref{C:quadratic} to
the weighted cross products gives
\begin{equation}\label{C:directionvariance}
 \operatorname{Var}(\nu'G'\eta/n)+\operatorname{Var}(\nu'G'f/n)
 \leq C\|\nu\|^2/n=O(T^{-1/2}).
\end{equation}
Both cross products have mean zero; geometric summation gives $\|\mu\|^2\asymp k$ and
$\|\gamma\|^2=O(k)$.

We next derive the bounds in~\eqref{C:auxbounds}
and the coefficient-sum limit in~\eqref{C:innovationlimits}
via a rank-one decomposition of $Q_0$.
Since $Z_0=G-\eta\mu'$, projection of $G$ onto the orthogonal complement
of $\eta$ gives a rank-one decomposition of $Q_0$. Define
\[
 M_{\eta}=I_n-\frac{\eta\eta'}{\eta'\eta},\qquad
 A=\frac{G'M_{\eta}G}{n},\qquad
 \mu_{\eta}=\mu-\frac{G'\eta}{\eta'\eta},\qquad
 s=\frac{G'M_{\eta}f}{n},
\]
where $M_{\eta}\in\mathbb R^{n\times n}$, $A\in\mathbb R^{k\times k}$,
and $\mu_{\eta},s\in\mathbb R^k$. Since
$Z_0=M_{\eta}G-\eta\mu_{\eta}'$ and $\eta'M_{\eta}G=0$, we have
\begin{equation}\label{C:rankone}
 Q_0=A+(\eta'\eta/n)\mu_{\eta}\mu_{\eta}',\qquad
 Z_0'r_0/n=s+\{(-\theta_T)^{k+1}(\eta'\eta/n)-\eta'f/n\}\mu_{\eta}.
\end{equation}
Equations~\eqref{C:gram}--\eqref{C:cross} imply
\[
 \|A-I_k\|=\Op(T^{-1/8}),\qquad
 \|s\|+\|\mu_{\eta}-\mu\|=\Op(T^{-1/4}).
\]
Consequently, $A$ and $Q_0$ are positive definite with probability
tending to one, their inverse norms are $\Op(1)$, and
$\|A^{-1}-I_k\|=\Op(T^{-1/8})$. The rank-one inverse formula yields
\begin{equation}\label{C:inverse}
 Q_0^{-1}=A^{-1}
 -\frac{(\eta'\eta/n)A^{-1}\mu_{\eta}\mu_{\eta}'A^{-1}}
 {1+(\eta'\eta/n)\mu_{\eta}'A^{-1}\mu_{\eta}},
\end{equation}
and therefore
\begin{equation}\label{C:alpha}
 \alpha_0=A^{-1}s+A^{-1}\mu_{\eta}
 \frac{(-\theta_T)^{k+1}(\eta'\eta/n)-\eta'f/n-(\eta'\eta/n)\mu_{\eta}'A^{-1}s}
 {1+(\eta'\eta/n)\mu_{\eta}'A^{-1}\mu_{\eta}}.
\end{equation}
For $\nu\in\{\mu,\gamma\}$, \eqref{C:directionvariance} implies
\[
 \nu's=\frac{\nu'G'f}{n}
 -\frac{(\nu'G'\eta/n)(\eta'f/n)}{\eta'\eta/n}=\Op(T^{-1/4}).
\]
Since $\|\nu\|=O(T^{1/4})$,
\[
\begin{aligned}
 |\nu'(A^{-1}-I_k)s|&\leq\|\nu\|\|A^{-1}-I_k\|\|s\|
 =\Op(T^{-1/8}),\\
 \nu'A^{-1}s&=\nu's+\nu'(A^{-1}-I_k)s=\op(1).
\end{aligned}
\]
Moreover, $(\mu_{\eta}-\mu)'A^{-1}s=\Op(T^{-1/2})$. Also,
\[
\begin{aligned}
 |\mu_{\eta}'A^{-1}\mu_{\eta}-\|\mu\|^2|
 &\leq\|A^{-1}-I_k\|\|\mu_{\eta}\|^2
 +\|\mu_{\eta}-\mu\|(\|\mu_{\eta}\|+\|\mu\|)=\op(k),\\
 |\nu'A^{-1}\mu_{\eta}-\nu'\mu|
 &\leq\|\nu\|\|A^{-1}-I_k\|\|\mu_{\eta}\|
 +\|\nu\|\|\mu_{\eta}-\mu\|=\op(k).
\end{aligned}
\]
The fraction in \eqref{C:alpha} has numerator
$(-\theta_T)^{k+1}+\op(1)$ and denominator
$(1+\|\mu\|^2)\{1+\op(1)\}$, of order $k$. Hence
\begin{equation}\label{C:direction}
 \nu'\alpha_0=(-\theta_T)^{k+1}
 \frac{\nu'\mu}{1+\|\mu\|^2}+\op(1),\quad \nu\in\{\mu,\gamma\},
 \qquad \|\alpha_0\|=\Op(T^{-1/4}).
\end{equation}
In particular, the case $\nu=\gamma$ is the coefficient-sum limit
in~\eqref{C:innovationlimits}, and
$\|\alpha_0\|=\Op(T^{-1/4})$ is the norm bound in~\eqref{C:auxbounds}.
For the first coordinate, \eqref{C:cross} gives
\[
\begin{aligned}
 \mathbf e_1'A^{-1}s&=\mathbf e_1's+\mathbf e_1'(A^{-1}-I_k)s
 =\Op(T^{-1/2})+\Op(T^{-3/8}),\\
 \mathbf e_1'A^{-1}\mu_{\eta}&=\mathbf e_1'\mu_{\eta}+\mathbf e_1'(A^{-1}-I_k)\mu_{\eta}
 =\Op(1)+\Op(T^{1/8}).
\end{aligned}
\]
Equation~\eqref{C:alpha} implies $\mathbf e_1'\alpha_0=\Op(T^{-3/8})$,
which is the leading-entry bound.
The correction to the first diagonal entry in \eqref{C:inverse} is
$\Op(T^{1/4}/k)=\op(1)$, completing the diagonal limit and thereby
all of~\eqref{C:auxbounds}.

It remains to prove the fitted-value limit
in~\eqref{C:innovationlimits}.
Equation~\eqref{C:rankone} implies
\[
 \|Z_0'r_0/n-(-\theta_T)^{k+1}\mu\|=\Op(T^{-1/4}).
\]
Multiplication by $\alpha_0$, followed by \eqref{C:direction} with
$\nu=\mu$, proves the fitted-value limit; the coefficient-sum limit
was established in~\eqref{C:direction} with $\nu=\gamma$.
\end{proof}

We next compare the regression of $r=(r_{k+2},\ldots,r_T)'$ on $Z$
with the regression of $r_0$ on $Z_0$. Recall the entries $p^c_{t,j}$
and $r^d_t$ in \eqref{C:perturbations} and \eqref{C:responseparts}.
Let $P^c\in\mathbb R^{n\times k}$ have row
$(p^c_{t,1},\ldots,p^c_{t,k})$ at date $t$, and set
$R=(R_1,\ldots,R_k)'$, $P^d=-\hat\psi_1\mathbf1_nR'$ and
$P=P^c+P^d$. For the response, define the $n$-vectors
\[
 r^c=-\theta_T(p^c_{k+2,1},\ldots,p^c_{T,1})',\qquad
 r^d=(r^d_{k+2},\ldots,r^d_T)'.
\]
Equations~\eqref{C:lagidentity} and \eqref{C:responseparts} give
\[
 Z=Z_0+P,\qquad r=r_0+r^c+r^d.
\]
For the lagged-level vector
$\ell=(\ty_{k+1},\ldots,\ty_{T-1})'\in\mathbb R^n$, define
\[
 Q=\frac{Z'Z}{n}\in\mathbb R^{k\times k},\qquad
 \alpha=Q^{-1}\frac{Z'r}{n}\in\mathbb R^k,\qquad
 w=\frac{Z'\ell}{n}\in\mathbb R^k.
\]
Thus $\alpha$ contains the coefficients from regressing $r$ on $Z$
alone. Parts (i) and (ii) of the next lemma compare this regression
with that of $r_0$ on $Z_0$. In particular, they imply
$\gamma'(\alpha-\alpha_0)=\op(1)$. Consequently, the first relation
in \eqref{C:innovationlimits} also holds with $\alpha$ in place of
$\alpha_0$. Part (iii) provides the moments needed to include
$\ell$ as an additional regressor.

\begin{lemma}\label{C:samplelemma}
Suppose model~\eqref{eq:model} and Assumption~\ref{ass:innovations} hold,
with fixed $\cbar<0$ and $\sigma_e^2=1$. Let $k=k_T$ be a deterministic
integer sequence such that $k/\sqrt T\to\kappa\in(0,\infty)$.
Then the following statements hold.
\begin{enumerate}
\renewcommand{\labelenumi}{(\roman{enumi})}
\item The matrices $Q,Q_0$ and responses $r,r_0$ satisfy
\begin{equation}\label{C:sample}
\begin{aligned}
 \|Q-Q_0\|&=\Op(T^{-1/2}),&
 \|(Z'r-Z_0'r_0)/n\|&=\Op(T^{-3/4}),\\
 \|r-r_0\|&=\Op(1).
\end{aligned}
\end{equation}
\item The matrix $Q$ is positive definite with probability tending to one,
and
\begin{equation}\label{C:transferbounds}
\begin{aligned}
 \|Q^{-1}\|&=\Op(1),&
 \|Q^{-1}-Q_0^{-1}\|&=\Op(T^{-1/2}),\\
 \|\alpha-\alpha_0\|&=\Op(T^{-3/4}).
\end{aligned}
\end{equation}
\item For $x_T=\xi'\xi/n$ from Lemma~\ref{C:glsinputs},
\begin{equation}\label{C:level}
\begin{aligned}
 \ell'\ell/n&=1+x_T+\op(1),&
 \|w-\mathbf e_1\|&=\Op(T^{-1/4}),\\
 \ell'r/n&=\Op(T^{-1/2}).
\end{aligned}
\end{equation}
\end{enumerate}
\end{lemma}

\begin{proof}
For (i), we first bound the products of $P^c$ and $r^c$ with $G$,
$\eta$ and $f$. Recall from \eqref{C:perturbations} that
$p^c_{t,j}=(c/T)\sum_{m=0}^{k-j}(-\theta_T)^m u_{t-j-m-1}$.
The identity $u_t=e_t+(\rho_T+\theta_T)X_{t-1}$ implies
$u_t=\sum_{b=1}^t\phi_{t-b}e_b$, where
\[
 \phi_l=
 \begin{cases}
  0, & l<0,\\
  1, & l=0,\\
  (\rho_T+\theta_T)\rho_T^{l-1}, & l\geq1.
 \end{cases}
\]
For fixed $j\in\{1,\ldots,k\}$, substitution gives
\[
 p^c_{t,j}=\sum_{b=1}^T a_{tb}e_b,\qquad
 a_{tb}=\frac cT\sum_{m=0}^{k-j}(-\theta_T)^m
             \phi_{t-j-m-1-b}.
\]
Here $a_{tb}$ is the coefficient of $e_b$ in $p^c_{t,j}$, with $j$
held fixed. All ensuing bounds are uniform in $j$. Since
$\sup_{1\leq t\leq T}\E u_t^2\leq C$ and
$|\rho_T+\theta_T|=O(T^{-1/2})$,
\[
\begin{aligned}
 \sum_{b=1}^T a_{tb}^2=\E(p^c_{t,j})^2&\leq C/T,\\
 \max_{1\leq b\leq T}|a_{tb}|
 &\leq(C/T)\{1+k|\rho_T+\theta_T|\}\leq C/T.
\end{aligned}
\]
For fixed $0\leq l\leq k+1$, define $D\in\mathbb R^{T\times T}$ by
\[
 D_{pq}=\frac1n\sum_{t=k+2}^T\mathbf1\{p=t-l\}\,a_{tq},
 \qquad 1\leq p,q\leq T.
\]
With $E=(e_1,\ldots,e_T)'$, the product
$n^{-1}\sum_{t=k+2}^T e_{t-l}p^c_{t,j}$ equals $E'DE$.
Distinct sample rows have distinct dates $t-l$, so
$\|D\|_{\mathrm F}^2\leq C/(nT)$ and
$|\operatorname{tr}D|\leq C/T$. The quadratic-form bound
\eqref{C:quadratic} therefore gives
\[
 \E\left|\frac1n\sum_{t=k+2}^T e_{t-l}p^c_{t,j}\right|^2
 \leq C/T^2,\qquad 0\leq l\leq k+1,\quad1\leq j\leq k.
\]
Recall that $G$ has entries $e_{t-j}$, while $\eta$ and $f$ have
entries $e_{t-k-1}$ and $e_t$, respectively. Summing the preceding
bound over the $k^2$ matrix entries or $k$ vector entries, and using
$r^c=-\theta_TP^c\mathbf e_1$, yields
\begin{equation}\label{C:localbounds}
\begin{aligned}
 \|G'P^c/n\|&=\Op(T^{-1/2}),\\
 \|(P^c)'\eta/n\|+\|(P^c)'f/n\|+\|G'r^c/n\|
 &=\Op(T^{-3/4}),\\
 \eta'r^c/n&=\Op(T^{-1}).
\end{aligned}
\end{equation}
The bound $\E(p^c_{t,j})^2\leq C/T$ also gives
$\|P^c\|_{\mathrm F}=\Op(T^{1/4})$ and $\|r^c\|=\Op(1)$.

We next obtain the corresponding bounds for
$P^d=-\hat\psi_1\mathbf1_nR'$ and $r^d$.
Lemma~\ref{C:glsinputs} and $\sum_{j=1}^k R_j^2=O(T^{3/2})$ imply
$\|P^d\|_{\mathrm F}=\Op(T^{1/4})$.
Recall from \eqref{C:responseparts} that
\[
 r^d_t=\frac cT\{\hat\psi_0+(t-1)\hat\psi_1\}
       -\hat\psi_1R_0.
\]
This expression is affine in $t$, with constant coefficient
$\Op(T^{-1/2})$ and slope $\Op(T^{-2})$; hence $\|r^d\|=\Op(1)$.
For $\tau=(k+2,\ldots,T)'$, independence of the innovations gives
\begin{equation}\label{C:lineproducts}
\begin{aligned}
 \|G'\mathbf1_n/n\|&=\Op(T^{-1/4}),&
 \|G'\tau/n\|&=\Op(T^{3/4}),\\
 |\eta'\mathbf1_n/n|+|f'\mathbf1_n/n|&=\Op(T^{-1/2}),&
 |\eta'\tau/n|+|f'\tau/n|&=\Op(T^{1/2}).
\end{aligned}
\end{equation}
Multiplication by $\hat\psi_1$, $R$, and the coefficients in $r^d_t$
establishes the counterparts of \eqref{C:localbounds} with $P^d,r^d$
in place of $P^c,r^c$. This step does not require the fitted
coefficients to be independent of the innovations.

Combining these bounds using $P=P^c+P^d$ and $r-r_0=r^c+r^d$ gives
$\|P\|_{\mathrm F}=\Op(T^{1/4})$, $\|r-r_0\|=\Op(1)$, and
\begin{equation}\label{C:perturbbounds}
\begin{aligned}
 \|G'P/n\|&=\Op(T^{-1/2}),\\
 \|P'\eta/n\|+\|P'f/n\|+\|G'(r-r_0)/n\|
 &=\Op(T^{-3/4}),\\
 \eta'(r-r_0)/n&=\Op(T^{-1}).
\end{aligned}
\end{equation}
This proves the third bound in \eqref{C:sample}. To obtain the first
two bounds there, use $Z=Z_0+P$ to write
\[
\begin{aligned}
 Z'Z-Z_0'Z_0&=Z_0'P+P'Z_0+P'P,\\
 Z'r-Z_0'r_0&=Z_0'(r-r_0)+P'r_0+P'(r-r_0).
\end{aligned}
\]
Since $Z_0=G-\eta\mu'$ and $r_0=f-(-\theta_T)^{k+1}\eta$,
\eqref{C:perturbbounds} and $\|\mu\|=O(T^{1/4})$ imply
\[
 \|Z_0'P/n\|=\Op(T^{-1/2}),\qquad
 \|Z_0'(r-r_0)/n\|+\|P'r_0/n\|=\Op(T^{-3/4}).
\]
The other terms satisfy
\[
\begin{aligned}
 \|P'P/n\|&\leq\|P\|_{\mathrm F}^2/n=\Op(T^{-1/2}),\\
 \|P'(r-r_0)/n\|&\leq\|P\|_{\mathrm F}\|r-r_0\|/n
 =\Op(T^{-3/4}).
\end{aligned}
\]
Substitution proves the bounds on $Q-Q_0$ and
$(Z'r-Z_0'r_0)/n$, the first and second bounds in
\eqref{C:sample}, respectively, completing (i).

For (ii), we first establish positive definiteness and the bounds
on $Q^{-1}$ and $Q^{-1}-Q_0^{-1}$.
Lemma~\ref{C:innovation} gives positive definiteness of $Q_0$
with probability tending to one and $\|Q_0^{-1}\|=\Op(1)$.
Together with $\|Q-Q_0\|=\op(1)$ from (i), these imply positive
definiteness of $Q$ with probability tending to one and
$\|Q^{-1}\|=\Op(1)$, the first bound in \eqref{C:transferbounds}.
The identity
\[
 Q^{-1}-Q_0^{-1}=Q^{-1}(Q_0-Q)Q_0^{-1}
\]
then gives $\|Q^{-1}-Q_0^{-1}\|=\Op(T^{-1/2})$, the second bound
in \eqref{C:transferbounds}.
For the coefficient difference, recall that
$Q\alpha=Z'r/n$ and $Q_0\alpha_0=Z_0'r_0/n$. Thus
\begin{equation}\label{C:alphadiff}
 \alpha-\alpha_0
 =Q^{-1}\{(Z'r-Z_0'r_0)/n-(Q-Q_0)\alpha_0\}
 =\Op(T^{-3/4}),
\end{equation}
where the final bound uses (i) and
$\|\alpha_0\|=\Op(T^{-1/4})$ from Lemma~\ref{C:innovation}.
This proves the third bound in \eqref{C:transferbounds}, completing
(ii). Since $\|\gamma\|=O(T^{1/4})$, it also gives
$\gamma'(\alpha-\alpha_0)=\Op(T^{-1/2})$.

For (iii), recall from \eqref{C:leveldecomp} that
$\ell=G\mathbf e_1+\xi$, where
$\xi_t=(\rho_T+\theta_T)X_{t-1}-\hat\psi_0-t\hat\psi_1$ and
$\xi=(\xi_{k+1},\ldots,\xi_{T-1})'$.
We first bound $G'\xi/n$, $\eta'\xi/n$ and $f'\xi/n$; these
products enter each of the three conclusions in (iii).
The coefficient of $e_b$ in $(\rho_T+\theta_T)X_{t-2}$ is
\[
 (\rho_T+\theta_T)\rho_T^{t-2-b}
 \mathbf1\{1\leq b\leq t-2\},
\]
and the sum of its squares is bounded uniformly in $t$.
Apply the construction of $D$ in (i) with these coefficients to
$n^{-1}\sum_{t=k+2}^T e_{t-j}(\rho_T+\theta_T)X_{t-2}$.
The resulting matrix has squared Frobenius norm $O(T^{-1})$.
Its trace is zero for $j=1$ and
$(\rho_T+\theta_T)\rho_T^{j-2}=O(T^{-1/2})$ for $j\geq2$.
For the products with $\eta$ and $f$, replace $e_{t-j}$ by
$e_{t-k-1}$ and $e_t$, respectively; the traces become
$(\rho_T+\theta_T)\rho_T^{k-1}$ and zero.
The remaining terms in $\xi_{t-1}$ are
$-\hat\psi_0-(t-1)\hat\psi_1$. Their products with $G$, $\eta$ and
$f$ are bounded by Lemma~\ref{C:glsinputs} and
\eqref{C:lineproducts}. Together with \eqref{C:quadratic}, this gives
\begin{equation}\label{C:Kproducts}
\begin{aligned}
 \mathbf e_1'G'\xi/n&=\Op(T^{-1/2}),&
 \|G'\xi/n\|&=\Op(T^{-1/4}),\\
 |\eta'\xi/n|+|f'\xi/n|&=\Op(T^{-1/2}).
\end{aligned}
\end{equation}

To obtain the expansion of $\ell'\ell/n$, use $\xi'\xi/n=x_T$
and $\|G\mathbf e_1\|^2/n=1+\Op(T^{-1/2})$ in
\[
 \ell'\ell/n
 =\|G\mathbf e_1\|^2/n+2\mathbf e_1'G'\xi/n+x_T
 =1+x_T+\op(1).
\]
This is the first relation in \eqref{C:level}, and it also implies
$\|\ell\|=\Op(\sqrt T)$.

For the bound on $w-\mathbf e_1$, independence gives
\[
 \E\|(G'G/n-I_k)\mathbf e_1\|^2\leq Ck/n,\qquad
 \mathbf e_1'G'\eta/n=\Op(T^{-1/2}).
\]
Substituting $Z=G-\eta\mu'+P$ and $\ell=G\mathbf e_1+\xi$ into
$w=Z'\ell/n$ yields
\[
 w=\frac{G'G}{n}\mathbf e_1+\frac{G'\xi}{n}
 -\mu\left(\frac{\eta'G\mathbf e_1}{n}
              +\frac{\eta'\xi}{n}\right)
 +\frac{P'\ell}{n}.
\]
Equation~\eqref{C:Kproducts}, $\|\mu\|=O(T^{1/4})$, and
\[
 \|P'\ell/n\|\leq\|P\|_{\mathrm F}\|\ell\|/n=\Op(T^{-1/4})
\]
therefore give $\|w-\mathbf e_1\|=\Op(T^{-1/4})$, the second
relation in \eqref{C:level}.

Finally, to bound $\ell'r/n$, first substitute
$\ell=G\mathbf e_1+\xi$ and
$r_0=f-(-\theta_T)^{k+1}\eta$ into $\ell'r_0/n$.
Equations~\eqref{C:cross} and \eqref{C:Kproducts} give
$\ell'r_0/n=\Op(T^{-1/2})$. Part (i) then implies
\[
 |\ell'(r-r_0)/n|
 \leq\|\ell\|\|r-r_0\|/n=\Op(T^{-1/2}).
\]
Hence $\ell'r/n=\Op(T^{-1/2})$, the third relation in
\eqref{C:level}, completing (iii).
\end{proof}

\begin{proof}[Proof of Proposition~\ref{prop:feasible_s2ar}]
We first establish the limit of $(1-\hat b(1))/\sqrt T$ in
\eqref{C:targets}. By \eqref{C:coeffmap}, it suffices to obtain
the limit of $\gamma'\hat\beta$.
Lemma~\ref{C:samplelemma} concerns the regression of $r$ on $Z$
alone; we now include the lagged level $\ell$.
Recall $Q=Z'Z/n$, $\alpha=Q^{-1}Z'r/n$ and $w=Z'\ell/n$.
For the response vector
$d=(\Delta\ty_{k+2},\ldots,\Delta\ty_T)'\in\mathbb R^n$,
\eqref{C:response} gives $d=(c/T)\ell+\theta_TZ\mathbf e_1+r$.
Thus the coefficients on $Z$ in the regression of $r$ on
$(\ell,Z)$ are $\hat\beta-\theta_T\mathbf e_1$, and the normal
equations are
\[
 \begin{pmatrix}\ell'\ell/n&w'\\w&Q\end{pmatrix}
 \begin{pmatrix}\hat a\\\hat\beta-\theta_T\mathbf e_1\end{pmatrix}
 =\begin{pmatrix}\ell'r/n\\Z'r/n\end{pmatrix},
 \qquad \hat a=\hat b_0-c/T.
\]
Solving this system gives
\begin{equation}\label{C:fwl}
 \hat\beta=\theta_T\mathbf e_1+\alpha+
 Q^{-1}w\frac{w'\alpha-\ell'r/n}{\ell'\ell/n-w'Q^{-1}w}.
\end{equation}

We show that the final term in \eqref{C:fwl}, after premultiplication
by $\gamma'$, is $\op(1)$. Lemma~\ref{C:innovation} and
Lemma~\ref{C:samplelemma}(ii)--(iii) give
\[
\begin{aligned}
 w'\alpha
 &=\mathbf e_1'\alpha_0+(w-\mathbf e_1)'\alpha_0
   +w'(\alpha-\alpha_0)=\Op(T^{-3/8}),\\
 \ell'r/n&=\Op(T^{-1/2}),\\
 w'Q^{-1}w
 &=\mathbf e_1'Q_0^{-1}\mathbf e_1+\op(1)=1+\op(1).
\end{aligned}
\]
Hence the numerator $w'\alpha-\ell'r/n$ is $\Op(T^{-3/8})$.
By Lemma~\ref{C:samplelemma}(iii), the denominator is
$x_T+\op(1)$. Lemma~\ref{C:glsinputs} gives
$x_T\Rightarrow x$ with $\Pr(x>0)=1$, so the denominator has
reciprocal $\Op(1)$. Together with the positive definiteness of $Q$
from Lemma~\ref{C:samplelemma}(ii), this also implies that
$(\ell,Z)$ has full column rank with probability tending to one.
Since
\[
 |\gamma'Q^{-1}w|
 \leq\|\gamma\|\|Q^{-1}\|\|w\|=\Op(T^{1/4}),
\]
the contribution of the final term in \eqref{C:fwl} to
$\gamma'\hat\beta$ is $\Op(T^{-1/8})=\op(1)$.
Combining the first relation in \eqref{C:innovationlimits} with
the coefficient bound in
Lemma~\ref{C:samplelemma}(ii) therefore gives
\begin{equation}\label{C:sumdirection}
 \gamma'\hat\beta=\theta_T\gamma'\mathbf e_1+
 (-\theta_T)^{k+1}\frac{\gamma'\mu}{1+\|\mu\|^2}+\op(1).
\end{equation}
For $\mu_j=(-\theta_T)^{k-j+1}$ and $\gamma=\mathbf1_k-\mu$,
geometric summation gives
\[
 (-\theta_T)^k\to e^{-\kappa\delta},\qquad
 \frac{\gamma'\mu}{1+\|\mu\|^2}\to\tanh(\kappa\delta/2).
\]
Since $\gamma'\mathbf e_1=1-(-\theta_T)^k$, substitution in
\eqref{C:sumdirection} yields
$\gamma'\hat\beta\xrightarrow{\IP}-\tanh(\kappa\delta/2)$.
Equation~\eqref{C:coeffmap} and $\sqrt T(1+\theta_T)=\delta$
then imply
\begin{equation}\label{C:coeflimit}
 \frac{1-\hat b(1)}{\sqrt T}
 =\frac1{\sqrt T}-\frac{\gamma'\hat\beta}{\delta}
 \xrightarrow{\IP}\frac1\delta\tanh(\kappa\delta/2).
\end{equation}

We next establish the residual-variance limit in \eqref{C:targets}.
Let $\operatorname{SSR}$ denote the residual sum of squares in the
autoregression with unit innovation variance. By the invertible
lag transformation and \eqref{C:response}, it equals the residual
sum of squares from regressing $r$ on $(\ell,Z)$. Successive
projection gives
\begin{equation}\label{C:ssr}
 \operatorname{SSR}/n=r'r/n-r'Z\alpha/n
 -\frac{(\ell'r/n-w'\alpha)^2}{\ell'\ell/n-w'Q^{-1}w}.
\end{equation}
We evaluate the three terms on the right-hand side in turn.
From \eqref{C:innovationarrays}, the law of large numbers and
\eqref{C:cross},
\[
 r_0'r_0/n=1+(-\theta_T)^{2k+2}+\op(1).
\]
Since $\|r-r_0\|=\Op(1)$ by Lemma~\ref{C:samplelemma}(i) and
$\|r_0\|=\Op(\sqrt T)$, this expansion also holds for $r'r/n$.

We next compare $r'Z\alpha/n$ with $r_0'Z_0\alpha_0/n$.
Equation~\eqref{C:rankone} gives
$\|Z_0'r_0/n\|=\Op(T^{1/4})$, while Lemma~\ref{C:innovation}
gives $\|\alpha_0\|=\Op(T^{-1/4})$.
Parts (i)--(ii) of Lemma~\ref{C:samplelemma} therefore imply
\[
 \left|\frac{r'Z\alpha-r_0'Z_0\alpha_0}{n}\right|
 \leq\left\|\frac{Z'r-Z_0'r_0}{n}\right\|\|\alpha\|
 +\left\|\frac{Z_0'r_0}{n}\right\|
  \|\alpha-\alpha_0\|=\op(1).
\]
The final quotient in \eqref{C:ssr} is $\Op(T^{-3/4})$ by the
numerator and denominator bounds established for \eqref{C:fwl}.
Substituting the second relation in \eqref{C:innovationlimits}
into \eqref{C:ssr} therefore gives
\[
 \operatorname{SSR}/n
 =1+\frac{(-\theta_T)^{2k+2}}{1+\|\mu\|^2}+\op(1)
 \xrightarrow{\IP}1,
\]
because $\|\mu\|^2\asymp k\to\infty$.
Restoring the innovation variance multiplies residual sums of
squares by $\sigma_e^2$ and leaves regression coefficients
unchanged. Thus
$\hat\sigma_{ek}^2
=\sigma_e^2(n/T)(\operatorname{SSR}/n)
\xrightarrow{\IP}\sigma_e^2$.
This establishes both limits in \eqref{C:targets}. Together with
\eqref{eq:s^2_ar} and \eqref{eq:lrv}, it proves
\[
 \frac{s_{AR}^2}{\omega_T^2}
 =\frac{\hat\sigma_{ek}^2}
 {\sigma_e^2\delta^2\{(1-\hat b(1))/\sqrt T\}^2}
 \xrightarrow{\IP}\coth^2(\kappa\delta/2).
\]
\end{proof}

\endgroup

\clearpage
\bibliographystyle{apacite}
\bibliography{references}
\end{document}